\documentclass[DIV=classic,11pt,a4paper]{article}

\usepackage[latin1]{inputenc}
\usepackage[ngerman,english]{babel}
\usepackage{lipsum}
\usepackage{amsmath,amssymb,amsthm}
\usepackage{graphicx}
\usepackage{caption}
\usepackage{subcaption}
\usepackage{comment}
\usepackage{
           ,bbm
           ,enumerate
           ,mathtools
           }
\usepackage[round,authoryear]{natbib}
\usepackage[usenames,dvipsnames]{color}
\usepackage[left=2.5cm,top=3.5cm,right=1.7cm,bottom=3cm]{geometry}
\usepackage{MnSymbol}
\newtheorem{theorem}{Theorem}[section]
\newtheorem{proposition}[theorem]{Proposition}
\newtheorem{corollary}[theorem]{Corollary}
\newtheorem{definition}[theorem]{Definition}
\newtheorem{lemma}[theorem]{Lemma}
\newtheorem{remark}[theorem]{Remark}
\newtheorem{example}[theorem]{Example}

\definecolor{darkgreen}{rgb}{0.00,0.40,0.00}

\numberwithin{equation}{section}

\newcommand{\1}{\mathbbm{1}}

\newcommand{\IE}{\mathbb{E}}

\newcommand{\IN}{\mathbb{N}}

\newcommand{\IP}{\mathbb{P}}

\newcommand{\IR}{\mathbb{R}}
\newcommand{\R}{\mathbb{R}}

\newcommand{\al}{\alpha}
\newcommand{\si}{\sigma}
\newcommand{\wh}{\widehat}
\newcommand{\ov}{\overline}

\newcommand{\cF}{\mathcal{F}}

\newcommand{\cP}{\mathcal{P}}
\newcommand{\cQ}{\mathcal{Q}}

\newcommand{\ud}{\mathrm{d}}

\newcommand{\udu}{\mathrm{d}u}

\newcommand{\essinf}{\mathrm{essinf}}

\newcommand{\var}[1]{\mathrm{VaR}^\alpha_{#1}}
\newcommand{\varu}[1]{\mathrm{VaR}^u_{#1}}
\newcommand{\avar}[1]{\mathrm{AVaR}^\alpha_{#1}}
\newcommand{\tildevar}[1]{\widetilde{\mathrm{VaR}}^\alpha_{#1}}
\newcommand{\tildeavar}[1]{\widetilde{\mathrm{AVaR}}^\alpha_{#1}}
\newcommand{\avarlow}[1]{\underline{\mathrm{AVaR}}^\alpha_{#1}}
\newcommand{\avarup}[1]{\overline{\mathrm{AVaR}}^\alpha_{#1}}

\newcommand{\GARCH}{\rm{GARCH}}

\def\balign#1\ealign{\begin{align}#1\end{align}}
\def\baligns#1\ealigns{\begin{align*}#1\end{align*}}

\newcommand{\halmos}{\quad\hfill\mbox{$\Box$}}

\allowdisplaybreaks[4]

\title{\bf Time-consistency of risk measures with GARCH volatilities and their estimation}

\author{Claudia Kl\"uppelberg\thanks{ Technische Universit\"at M\"unchen, Zentrum Mathematik, Boltzmannstra\ss e 3, 85748 Garching, Germany, \textsf{j.zhang@tum.de\,,\,cklu@tum.de} }
 \and
Jianing Zhang\footnotemark[1]
} 

\begin{document}
\selectlanguage{english}
\maketitle

\begin{abstract}
\noindent In this paper we study time-consistent risk measures for returns that are given by a GARCH$(1,1)$ model. We present a construction of risk measures based on their static counterparts that overcomes the lack of time-consistency. We then study in detail our construction for the risk measures Value-at-Risk (VaR) and Average Value-at-Risk (AVaR). While in the VaR case we can derive an analytical formula for its time-consistent counterpart, in the AVaR case we derive lower and upper bounds to its time-consistent version. Furthermore, we incorporate techniques from Extreme Value Theory (EVT) to allow for a more tail-geared statistical analysis of the corresponding risk measures. We conclude with an application of our results to a data set of stock prices. 
\end{abstract}


\smallskip

{\bf 2010 JEL classification:} C02, C22, C58, G17, G32

\smallskip

\medskip

\section{Introduction}\label{sec:intro}
In the wake of the financial crisis risk management constitutes a constant active field that attracts both mathematical research and quantitative requirements for the practical implementation. Most financial institutions need to abide with the Basel II/III accords that prescribe certain risk management rules to be applied to internal risk control and that are under periodic regulatory supervision. Over the last two decades the key notion of risk management arose in the form of a risk measure referred to as Value-at-Risk (VaR). Simply put, VaR determines the risk capital of a financial institution as the quantile of a profit-and-loss distribution with respect to some prescribed (either by regulation or by internal rules) time horizon and confidence level. An axiomatic approach to the field of risk measures is given by \citet{artzneretal} in which the notion of the \emph{coherent} risk measure is introduced and where it has been realized that VaR does not always satisfy the property of coherence. \citet{
artzneretal} introduce a risk measure that amends the lack of coherence that is nowadays known as the Average-Value-at-Risk (AVaR). 
An extension to \emph{convex} risk measures is given in \citet{foellmerschied_convex}, which integrates existing notions of risk into the mathematical framework of convex dual theory and, hence, allows for deep and powerful dual characterizations. In order to account for the dynamic stochastic evolution of profit-and-loss positions the static risk measurement has been extended to the class of \emph{dynamic risk measures}, which treats the risk measure not only as a (nonlinear) expectation but as a stochastic process, see e.g. \citet{detlefsen} and \citet{riedel} for the extension to the dynamic setting by means of convex dual theory. 
It has been realized in this dynamic framework that most existing static risk measures do not transfer in a straightforward manner into processes without violating the required property of \emph{time-consistency}. 
A time-consistent dynamic risk measure secures the consistent behavior of a risk measure that, if a portfolio is riskier than another portfolio at some future time, then this portfolio has been riskier that the other portfolio at any time before. 
The literature on time-consistency of risk measures is diverse and rich as different mathematical viewpoints can be adopted to prevent the consistency property. 
An incomplete chronicle of research done in the field of time-consistent risk measures includes \citet{pengnonlinear}, \citet{riedel}, \citet{detlefsen}, \citet{Weber}, \citet{foellmerpenner}, \citet{RoordaSchumacher}, \citet{penner}, \citet{bion}, and \citet{BCP}. A major result from the research on time-consistency reveals that in the class of law-invariant risk measures there is only one risk measure that, upon transfer into a time-dynamic process setting, supports time-consistency, namely the entropic risk measure (cf. \citet{foellmerknispel}).

\medskip

In parallel to the aforementioned theoretical work statistical models and methods have been developed to calibrate and integrate risk measures to real world data. As the industry standard VaR and its coherent counterpart AVaR are law-invariant risk measures, the main goal for an implementation of (A)VaR is to find a good estimate of the profit-and-loss distribution in the relevant region. 
In this field, the major class of estimation methods comprise the historical simulation method, methods based on Gaussian distribution assumptions and methods based on Extreme Value Theory (EVT). 
We refer to \citet{mfe}, in particular Chapter~2 and Chapter~7, for a detailed account and references to methods of profit-and-loss distribution estimation. 
More background on extreme value theory can be found in the monograph \cite{ekm}.
\citet{mcneil_frey} propose an implementation of VaR and AVaR that is based on an estimation of the log-returns distribution using a combination of  a GARCH$(1,1)$ model fit and an EVT approach for the residuals.
Their method proceeds in a two-step scheme: first, the GARCH$(1,1)$ model mimics the inherent stochastic volatility of financial time series, and the GARCH parameters are estimated by a pseudo maximum likelihod method.
Second, they adopt a Peaks-over-Threshold (POT) approach to the residuals and only consider those residuals that exceed a critical value. 
The POT  method justifies fitting a Generalized Pareto distribution (GPD) by means of a maximum likelihood method  (e.g. \cite{ekm}, Section~3.4 and Section~6.5) 
It is also in accord with the typically high confidence levels that are imposed on (A)VaR to zoom into the extreme branch of losses. 
Applying the POT method to the residuals rather than directly to the log-returns has the advantage that the fitting procedure to the extremes only needs to be applied once due to the white noise property of the residuals. 
Using these two steps, \citet{mcneil_frey} succeed to estimate (A)VaR by fitting a distribution that adequately accounts for the extremes in the tail and under mild conditions allows for closed form formulas for VaR and AVaR. 

\medskip

The goal of our paper is to incorporate dynamic time-consistency for VaR and AVaR. We investigate the extension of static risk measures to dynamic counterparts that satisfy time-consistency. A key property to succeed in this transfer is the dynamic programming principle, see \citet{cheridito_stadje}, \citet{ck}. 

\medskip

The two-step estimation scheme from \citet{mcneil_frey}  using GARCH$(1,1)$ and EVT 
allows us to derive a closed form expression for the dynamic time-consistent VaR that is easily implemented using the estimated GPD and the GARCH parameters. 
For AVaR however, such a closed form expression cannot be obtained and
we derive closed form lower and upper approximations to AVaR.
On top of being more conservative than their static counterparts, the dynamic time-consistent VaR offers the benefit that the risk measurement of aggregated losses, which in e.g. \citet{mcneil_frey} have to be estimated by simulation methods, can now be estimated in a (semi-)closed way by simply aggregating the VaRs of the single positions at different future time points. 

\medskip

The paper is structured as follows. In Section~\ref{sec:risk_measures} we present preliminaries on dynamic risk measures along with the dynamic programming principle characterization.
 Moreover, we introduce the GARCH$(1,1)$ loss model, which establishes the  model framework for the entire paper. 
In Section~\ref{sec:var} we apply the new methodology from the previous section to derive a closed form expression for the time-consistent VaR and investigate its properties concerning the evolution over time and prove the linearization of aggregated losses. 
Section~\ref{sec:avar} is devoted to the study of AVaR. Since a closed form expression for time-consistent AVaR is not possible, as an alternative, we derive closed form expressions for pragmatic bounds to AVaR and study the properties as in the previous section. The proofs of the results of Sections~3 and~4 are postponed to the Appendix. In the last Section~\ref{sec:evt} we give a rehash on the part of extreme value theory that is relevant for our purpose, and apply our results to a data set of stock prices. 

\section{Conditional risk measures}\label{sec:risk_measures}

Given a probability space $(\Omega,\cF,\IP)$ we consider a filtration $(\cF_t)_{t=0}^T$ where $T\in\IN$. We denote by $L^0(\cF_t)$ with $t\in\{0,\ldots,T\}$ the set of all $\cF_t$-measurable random variables $X: \Omega \to \IR$. In this paper, the space $L^0(\cF_T)$ represents the space of all financial positions for which we need a risk assessment. Typically, we will be interested in losses, i.e. the negatives of log-returns of financial data. 

Since conditional risk measures are random variables, all properties, equalities and inequalities below hold almost surely with respect to $\IP$, and we assume this throughout without making extra mention of it.

\begin{definition}\label{def:cond_rm}
For $t\in\{0,\ldots,T\}$ a family of mappings $( \phi_t )_{t=0}^T$ with $\phi_t: L^0(\cF_T) \to L^0(\cF_t)$ is a \emph{dynamic monetary risk measure} if it satisfies the following properties:
\begin{itemize}
\item[(i)] Normalization: $\phi_t(0)=0$ for $t=0,\ldots,T$;
\item[(ii)] Monotonicity: $\phi_t(X) \geq  \phi_t(Y)$ for all $X,Y\in L^0(\cF_T)$ such that $X\geq Y$, for $t=0,\ldots,T$;
\item[(iii)] Translation invariance: $\phi_t(X+m) = \phi_t(X) + m$ for all $X\in L^0(\cF_T)$ and $m\in L^0(\cF_t)$, for $t=0,\ldots,T$.
\end{itemize}
\end{definition}

If $L^0(\cF_T)$ represents the space of all profit and loss variables, the above definition leads to the notion of a \emph{dynamic monetary utility function}, see Definition 2.1 in \citet{ck}. If a dynamic monetary risk measure $\phi$ satisfies in addition to Definition \ref{def:cond_rm} (i)-(iii) 
\begin{itemize}
\item Positive homogeneity: $\phi_t(\lambda X) = \lambda \phi_t(X)$ for all $X \in L^0(\cF_T)$ and   $\lambda >0$, for $t=0,\ldots,T$; 
\item Subadditivity: $\phi_t(X+Y) \leq \phi_t(X) + \phi_t(Y)$ for all $X,Y \in L^0(\cF_T)$, for $t=0,\ldots,T$,
\end{itemize}
then we say that $\phi$ is a \emph{coherent} (dynamic monetary) risk measure.

\begin{definition}\label{def:time_con}
A dynamic monetary risk measure $\phi:= (\phi_t  )_{t=0}^T$ is \emph{time-consistent} if 
\baligns
\phi_{t+1}(X) \geq \phi_{t+1}(Y) ~\text{ implies  } ~ \phi_{t}(X) \geq \phi_{t}(Y), 
\ealigns
for all $X,Y \in L^0(\cF_T)$, for $t=0,\ldots,T-1$.
\end{definition} 

The following useful characterization of time-consistency can be found in \citet{ck}.

\begin{proposition}\label{prop:dpp}
A dynamic monetary risk measure $( \phi_t )_{t=0}^T$ is time-consistent if and only if it satisfies the Bellman principle
\balign\label{eq:dpp}
\phi_t(X) = \phi_t\big( \phi_{t+1}(X) \big)
\ealign
for all $X\in L^0(\cF_T)$, and $t=0,\ldots,T-1$. 
\end{proposition}

It has been noted in \citet{cheridito_stadje} and \citet{ck} that there is another way to construct time-consistent dynamic risk measures: let $(\rho_t)_{t=0}^{T-1}$ be an arbitrary dynamic monetary risk measure
\baligns
\rho_t: L^0(\cF_T) \to L^0(\cF_t), \quad t=0,\ldots,T-1,
\ealigns
then the backward iteration 
\balign\label{eq:backwd_ind}
\phi_T(X):= X, \quad \phi_t(X):= \rho_t\big(\phi_{t+1}(X)\big), \quad t=0,\ldots,T-1,
\ealign
defines a process $(\phi_t)_{t=0}^T$ which by definition is a time-consistent dynamic risk measure.
The following property is a straightforward consequence of the construction of $(\phi_t)_{t=0}^T$.

\begin{corollary}\label{cor:iterate}
For $X\in L^0(\cF_T)$ we have for $t=0,\cdots,T-1$
\balign
\phi_t(X) = \Big( \rho_t \circ \rho_{t+1} \circ \cdots \circ \rho_{T-1} \Big) (X). 
\ealign
\end{corollary}

For a coherent risk measure $\phi$, its subadditivity property implies that for any fixed $t\in\{0,\ldots,T\}$ and $m\in \IN$ such that $t+m\le T$ and any $X_{t+k} \in L^0(\cF_{t+k})$ for $k=1,\ldots,m$ we have
\balign\label{eq:linearize}
\phi_t\big( \sum_{k=1}^m X_{t+k} \big) \leq \sum_{k=1}^m \phi_t( X_{t+k} ).
\ealign

We construct our time-consistent dynamic risk measures by backwards iteration.

\subsection{The GARCH(1,1) model for loss positions}\label{sec:model_choice}

Recall that we are interested in the risk assessment of losses. 
The focus of this paper is on a particular class of loss processes $(L_t)_{t=0}^T$: its dynamics is governed by a GARCH(1,1) process and typically represent (negative) log-returns. 
It holds that $(L_t)_{t=1}^T$ satisfies
\begin{equation}\label{eq:garch_loss}
\begin{aligned}
&L_t = \sigma_t Z_t,\\
&\sigma_t^2 = a_0 + a_1 L_{t-1}^2 + b \sigma_{t-1}^2,
\end{aligned}
\end{equation}
where $a_0,a_1,b>0$ are the model parameters, $\sigma_0$ and $L_0$ are $\cF_0$-measurable initial random variables, and $(Z_t)_{t=1}^T$ is a strict white noise process (independently identically distributed with zero mean and unit variance). 
Note also that by \eqref{eq:garch_loss} $\sigma_t$ is measurable with respect to $\cF_{t-1}$ for every $t=1,\ldots,T$. 

\medskip

We denote by $F_Z:\IR \to [0,1]$ and $F_Z^{-1}:[0,1] \to \IR$ the  distribution function and the left-continuous quantile function of each $Z_t$,  respectively; i.e.,
\balign\label{leftinv}
F_Z(z) = \IP( Z_t \leq z ), \quad F^{-1}_Z(\alpha) = \inf\{ x\in\IR:  F_Z(x) \ge \alpha \}, \quad \alpha \in (0,1), \quad t=0,\ldots,T. 
\ealign

For properties of the quantile function $F_Z^{-1}$ we refer to \cite{Resnick}, Section~0.2, or \cite{ekm}, Proposition~A1.6.
 
We assume that $F_Z$ is strictly increasing, thus $F_Z^{-1}$ is continuous, and that the right endpoint of $Z_t$ is infinite; i.e.,
\[
x_F = \inf\{x\in\IR: F_Z(x) = 1\} = \infty.
\]
If necessary we identify $F^{-1}_Z(1)$ with $x_F=\infty$.
Since $Z$ has infinite right endpoint, and $\alpha$ is close to 1, $F_Z^{-1}(\alpha)$ is as a rule positive.
We shall also need the quantile function of $Z^2$ and note that for $Z$ symmetric, $F^{-1}_{Z^2}(\alpha)= F^{-1}_{Z}(\frac12(\alpha+1))^2$. 
Note further that $\alpha\le \frac12(\alpha+1)$ for $\alpha\in(0,1)$, hence (for $F^{-1}_{Z}(\alpha)>1$)
\balign\label{comp}
F^{-1}_{Z}(\alpha)\le F^{-1}_{Z}(\alpha)^2 \le F^{-1}_{Z}(\frac12(\alpha+1))^2 = F^{-1}_{Z^2}(\alpha).
\ealign

We summarize the assumptions which we will assume throughout the paper. \\

{\bf Assumptions A: }\quad 
We assume that $F_Z$ is strictly increasing with support $\R$ and that $F^{-1}_{Z}(\alpha)>0$.
For simplicity, we also assume that $Z$ is symmetric.\\

Since we often work with distribution tails, we note that $F^{-1}_Z$ can also be represented as
\balign\label{varZ}
 F^{-1}_Z(\alpha) = \inf\{ x\in\IR:  P(Z_t>x) \le 1-\alpha \}, \quad \alpha \in (0,1), \quad t=0,\ldots,T.
\ealign

\section{Conditional time-consistent Value-at-Risk}\label{sec:var}

In this section we study Value-at-Risk (VaR) in the framework of dynamic time-consistent risk measures. 
One typically considers $L \in L^0(\cF_T)$ which represents a possibly large loss position, for which the probability of $L$ exceeding a loss threshold $m>0$ should be bounded by a small probability $1-\alpha$, i.e. $\alpha$ is typically close to $1$. 
The smallest loss threshold $m$ which satisfies this bound is the $\var{}$. 
Several versions of the (conditional) VaR definition can be found in the literature.
In analogy to \eqref{varZ} we work throughout with the following, which caters best to the purpose of the treatments in this paper.  

\begin{definition}\label{def:var}
Given a loss position $L\in L^0(\cF_T)$ the \emph{Value-at-Risk at level $\alpha \in (0,1)$ at time $t\in\{0,\ldots,T\}$ for $L$} is defined by
\balign\label{eq:var}
\var{t}(L) := \essinf\big\{m \in L^0(\cF_t): ~ \IP( L \le m  ~|~\cF_t) \geq \alpha \big\}, 
\ealign
\end{definition}

\subsection{Time-consistent VaR for single day losses}\label{sec:var_single}

We start this section with the following example, which is the core object of interest in \citet{mcneil_frey}.

\begin{example}\label{ex:1_day}\rm
For $t=0,\ldots,T-1$ let $L_{t+1}$ be given by \eqref{eq:garch_loss}.  
Then $\var{t}(L_{t+1})$  is the \emph{1-day-ahead-\text{VaR}}, which can be computed straightforwardly as
\baligns
\var{t}(L_{t+1}) &= \essinf\big\{m \in L^0(\cF_t): ~ \IP( \sigma_{t+1}Z_{t+1} \le m  ~|~\cF_t) \geq \alpha \big\}.
\ealigns
Since $\sigma_{t+1}$ is $\cF_t$-measurable, we also have that $\tilde{m} := m/\sigma_{t+1}$ is $\cF_t$-measurable. Using the independence between $Z_{t+1}$ and $\cF_t$, and also \eqref{varZ}, we can continue 
\balign\label{varL}
\var{t}(L_{t+1}) &= \sigma_{t+1} \, \essinf\big\{ \tilde{m} \in L^0_+(\cF_t): ~ \IP( Z_{t+1} \le \tilde{m} ~|~\cF_t) \geq \alpha \big\}\nonumber\\
 &= \sigma_{t+1} \, \inf\big\{ \tilde{m} \in\IR_+: ~ \IP( Z_{t+1} \le \tilde{m}) \ge \alpha \big\}\nonumber\\
&= \sigma_{t+1} \, F_Z^{-1}(\alpha). 
\ealign
Moreover, since $L_t\in\cF_t$,
\balign\label{2sum}
\var{t}(L_t+L_{t+1}) = L_t +\var{t}(L_{t+1}) = L_t + \sigma_{t+1} F^{-1}_Z(\al).
\ealign
\halmos
\end{example}

There are examples showing that Value-at-Risk from Definition~\ref{def:var} is not time-consistent (e.g. \citet{cheridito_stadje} or \citet[Example 11.13]{foellmerschied}). 
As the GARCH$(1,1)$ model \eqref{eq:garch_loss} is defined by an iteration, one could hope that for this specific model $\var{}$ is time-consistent.
However, this is not true and we provide a counterexample, which makes use of Proposition~\ref{prop:dpp}.

\begin{example}\label{ex:counter_ex}\rm
In order to see why in the framework of GARCH$(1,1)$ losses VaR cannot be time-consistent, recall that according to Proposition~\ref{prop:dpp} $\var{}$ is time-consistent if and only if it satisfies the dynamic programming principle
\[
\var{t} = \var{t} \circ \var{t+1}, \quad t=0,\ldots,T-1.
\]
For $t\in\{0,\ldots,T-2\}$, by \eqref{varL}, we have $\var{t+1}(L_{t+2}) = \sigma_{t+2}F^{-1}_Z(\alpha)$ and, hence,
\[
\var{t} \big(\var{t+1}( L_{t+2} ) \big) = \var{t} ( \sigma_{t+2}F^{-1}_Z(\alpha) ).
\]
We compute $\var{t}( L_{t+2})$ and $\var{t} ( \sigma_{t+2}F^{-1}_Z(\alpha) )$ for the GARCH$(1,1)$ model: 
\baligns
\var{t} ( \sigma_{t+2}F^{-1}_Z(\alpha) ) = m^\ast &= \essinf\{m\in L^0(\cF_t): \IP\big(\sigma_{t+2}F^{-1}_Z(\alpha)\le m \mid \cF_t\big)\ge \al\}\\
 &= \essinf\{m\in L^0(\cF_t): 
 \IP\big( \sqrt{a_0+\sigma_{t+1}^2(a_1 Z_{t+1}^2+b)} F^{-1}_Z(\alpha)  \le {m}\mid \cF_t\big)\ge \al\}.
\ealigns
Since the function $\sqrt{(a_0+\sigma_{t+1}^2(a_1 Z_{t+1}^2+b)){F^{-1}_Z(\alpha)^2}}$ is strictly increasing in $Z_{t+1}^2$ and $\cF_t$-measurable, we obtain
\balign\label{var1}
m^\ast = \sqrt{(a_0 + \sigma_{t+1}^2 (a_1F^{-1}_{Z^2}(\alpha)  + b ))} F^{-1}_Z(\alpha).
\ealign
Next we compute 
\baligns
 \var{t}( L_{t+2} ) = m^{\ast\ast} &= \essinf\{m\in L^0(\cF_t): \IP\big(\si_{t+2} Z_{t+2} \le m \mid \cF_t\big)\ge \al\}\\
 &= \essinf\{m\in L^0(\cF_t): \IP\big(\sqrt{ a_0+\sigma_{t+1}^2(a_1 Z_{t+1}^2+b)} Z_{t+2} \le m \mid \cF_t\big)\ge \al\}.
\ealigns
Now assume that $m^{\ast\ast} = m^{\ast}$ for all $\al\in(0,1)$.
We denote by $\IP_t$ the conditional probability with respect to $\cF_t$ and calculate
\baligns
\alpha &=  \IP_t\big(\sqrt{ a_0+\sigma_{t+1}^2(a_1 Z_{t+1}^2+b)}  \le  \sqrt{a_0 +\sigma_{t+1}^2
 (a_1F^{-1}_{Z^2}(\alpha)  + b )} \big) \\
&= 
  \IP_t\big(\sqrt{ a_0+\sigma_{t+1}^2(a_1 Z_{t+1}^2+b)} Z_{t+2} \le  \sqrt{a_0 +\sigma_{t+1}^2 (a_1F^{-1}_{Z^2}(\alpha)  + b )}   F^{-1}_Z(\alpha) \big) \\
&= 2 \int_{0}^{\infty}  \IP_t\big(\sqrt{ a_0+\sigma_{t+1}^2(a_1 Z_{t+1}^2+b)} z \le  \sqrt{a_0 +\sigma_{t+1}^2 (a_1F^{-1}_{Z^2}(\alpha)  + b )} F^{-1}_Z(\alpha) \big)  dF_Z(z) -1\\
& =  2 \int_{0}^{F^{-1}_{Z}(\alpha) }  \IP_t\big(\sqrt{ a_0+\sigma_{t+1}^2(a_1 Z_{t+1}^2+b)} z \le  \sqrt{a_0 +\sigma_{t+1}^2 (a_1F^{-1}_{Z^2}(\alpha)  + b )} F^{-1}_Z(\alpha)  \big)  dF_Z(z)\\
& + 2 \int_{F^{-1}_{Z}(\alpha) }^\infty  \IP_t\big(\sqrt{ a_0+\sigma_{t+1}^2(a_1 Z_{t+1}^2+b)} z \le  \sqrt{a_0 +\sigma_{t+1}^2 (a_1F^{-1}_{Z^2}(\alpha)  + b )} F^{-1}_Z(\alpha)  \big)  dF_Z(z) -1
\ealigns
Now note for the first integral that ${F^{-1}_Z(\alpha)/z}$ decreases in $z$ to 0 and has minimum 1 over the integral range.
This implies for the probability under the integral, that the left-hand random variable scaled by $z$ decreases with $z$ to $\alpha$. Moreover, since the support of $Z_{t+1}$ has infinite right endpoint, the second integral is positive. 
Hence, we estimate the right-hand side by
\baligns
 \ge 2 \alpha^2 + 2a -1,
\ealigns
where $a>0$. However, for $\alpha$ close to 1 we have $\alpha+1< 2\alpha^2 +2a$.
\halmos
\end{example}

\medskip

\citet{cheridito_stadje} propose to amend the time-inconsistency of VaR using the backward iteration \eqref{eq:backwd_ind}. This gives rise to the following definition.

\begin{definition}\label{def:consistent_var}
Given a loss position $L\in L^0(\cF_T)$ and $\var{}$ from Definition~\ref{def:var}.
Then the  \emph{time-consistent Value-at-Risk at level $\alpha \in (0,1)$ for $L$} is defined by
\balign\label{eq:consistent_var}
\tildevar{T}(L) := \var{T}(L) = L, \quad \tildevar{t}(L) := \var{t}\big( \tildevar{t+1}(L) \big), \quad t=0,\ldots,T-1.
\ealign
\end{definition}

In the notation of the construction from the recursion \eqref{eq:backwd_ind}, this corresponds to $\rho_t := \var{t}$ and $\phi_t := \tildevar{t}$. 
As a consequence of the construction of $(\tildevar{T}(X))_{t=0}^T$ we find for $L\in L^0(\cF_T)$
\balign\label{cor:var_iterate}
\tildevar{t}(L) = \Big( \var{t} \circ \var{t+1} \circ \cdots \circ \var{T-1} \Big) (L). 
\ealign

\medskip

The choice of the GARCH$(1,1)$ model \eqref{eq:garch_loss} entails the convenient feature that the $m$-day ahead VaR assessment allows for a closed form solution. More precisely, we can derive an analytical solution for the time $t$ risk assessment of the GARCH(1,1) loss at terminal time $T$ as follows (as usual we set $\sum_{k=0}^{-1} a_n=0$). The proof is given in Appendix~\ref{AppA}.
 
\begin{theorem}\label{prop:garch_var_terminal}
Let $(L_t)_{t=0}^T$ be the loss process given by the \GARCH(1,1) model \eqref{eq:garch_loss}. 
Then we have
\balign\label{eq:garch_var}
\tildevar{t}(L_T) = F_Z^{-1}(\alpha) \sqrt{ \cP^T_{t}\big(a_1 F_{Z^2}^{-1}(\alpha) + b \big) }, \quad t=0,\ldots,T-1, 
\ealign
where $\cP^T_t: \IR \to \IR$ is an $\cF_t$-measurable mapping given by
\balign\label{eq:poly}
\cP^T_t(x) := a_0 \sum_{k=0}^{T-t-2} x^k + \sigma_{t+1}^2 x^{T-t-1}, \quad t=0,\ldots,T-1.
\ealign
\end{theorem}

\subsection{Time-consistent VaR for aggregated losses}\label{sec:var_multi}

We now come to the computation of the $m$-day-ahead $\var{}$.  So far, we have considered risk positions at a fixed day $T$ that is ahead of time $t<T$ up to which information in the form of the filtration $\cF_t$ is available. The $m$-day-ahead $\var{}$ is a risk assessment of aggregated losses $L_{t+k}$ from the time period $[t+1, t+m]$ for $t+m\le T$. 
Next we show that $\tildevar{}$ linearizes across the aggregation of GARCH$(1,1)$ losses. 

\begin{proposition}\label{prop:var_linear}
Let $(L_t)_{t=0}^T$ be the loss process given by the \GARCH(1,1) model \eqref{eq:garch_loss}, then we have for fixed $t\in\{1,\ldots,T-1\}$ and $m\in\IN$ such that $ t+m\le T$,
\balign\label{varsum}
\tildevar{t}\Big( \sum_{k=1}^{m} L_{t+k} \Big) = F_Z^{-1}(\al) \sum_{k=1}^m 
\sqrt{ \cP_{t}^{t+k}\big(a_1 F_{Z^2}^{-1}(\alpha) + b \big) }=\sum_{k=1}^m \tildevar{t}\big(L_{t+k} \big).
\ealign
\end{proposition}

\begin{proof}
First note that for $m=2$ we know from \eqref{2sum} that
\baligns
\tildevar{t}\big( L_{t+1} + L_{t+2} \big) &= \var{t} \big( L_{t+1} + \var{t+1}(L_{t+2}) \big)\\
&= \var{t} \big( \sigma_{t+1} Z_{t+1} + \sigma_{t+2} F_Z^{-1}(\alpha) \big),
\ealigns
where the second line follows from \eqref{varL}. 
By \eqref{eq:garch_loss} we have $\sigma_{t+2}^2 = a_0 + \sigma_{t+1}^2 (a_1  Z_{t+1}^2 + b)$ which transforms the last equation into
\baligns
\tildevar{t}\big( L_{t+1} + L_{t+2} \big) &=  \var{t} \big( \sigma_{t+1} Z_{t+1} + F_Z^{-1}(\alpha) \sqrt{a_0 +\sigma_{t+1}^2( a_1  Z_{t+1}^2 + b)}\big).
\ealigns
By the definition of $\var{t}$ and the fact that $\sigma_{t+1} Z_{t+1} + F_Z^{-1}(\alpha) \sqrt{a_0 +  \sigma_{t+1}^2 ( a_1 Z_{t+1}^2 + b)}$ is a strictly increasing $\cF_t$-measurable function of $Z^2_{t+1}$, we find that 
\balign\label{hack001}
\tildevar{t}\big( L_{t+1} + L_{t+2} \big) &=  \sigma_{t+1} F_Z^{-1}(\alpha) + F_Z^{-1}(\alpha) \sqrt{a_0 + \sigma_{t+1}^2( a_1 F_{Z^2}^{-1}(\alpha)^2 + b)},
\ealign
which is equal to the sum $\tildevar{t}\big(L_{t+1} \big) + \tildevar{t}\big(L_{t+2}\big)$ and also equal to the corresponding sum in the center.
We proceed by induction and assume that \eqref{varsum} is true for $\sum_{k=1}^{m-1}  L_{t+k}$.
Since the sum is $\cF_{t+m}$-measurable, we obtain by \eqref{cor:var_iterate}
\baligns
\tildevar{t}\Big( \sum_{k=1}^{m} L_{t+k} \Big) 
&= \Big( \var{t} \circ \var{t+1} \circ \cdots \circ \var{t+m-1} \Big)\Big( \sum_{k=1}^{m-1} L_{t+k}  + L_{t+m} \Big)\\
& = \Big( \var{t} \circ \var{t+1} \circ \cdots \circ \var{t+m-2} \Big)\Big( \sum_{k=1}^{m-1} L_{t+k}  + \var{t+m-1} (L_{t+m} )\Big)\\
& =  \var{t}(L_{t+1} + \var{t+1} (L_{t+2}+ \cdots +  \var{t+m-2}(L_{t+m-1} )  + \var{t+m-1} (L_{t+m} )))\\
& = \var{t}\Big(L_{t+1} + \tildevar{t+1}\Big(\sum_{k=2}^{m}L_{t+k}  \Big)\Big)
= \var{t}\Big(L_{t+1} +\sum_{k=2}^{m}  \tildevar{t+1}\Big(L_{t+k}  \Big)\Big)
\ealigns
where the last identity follows by the induction hypothesis, which also implies
\baligns
\sum_{k=2}^{m}  \tildevar{t+1}\Big(L_{t+k}  \Big) &=  F_Z^{-1}(\alpha) \sum_{k=2}^{m} \sqrt{  \cP_{t+1}^{t+k} \big( a_1 F_{Z^2}^{-1}(\alpha) + b \big)  } \\
&= F_Z^{-1}(\alpha)\sum_{k=2}^{m}  \sqrt{ a_0 \sum_{j=0}^{k-3} \big( a_1 F_{Z^2}^{-1}(\alpha) + b \big)^j + \sigma_{t+2}^2 \big( a_1 F_{Z^2}^{-1}(\alpha) + b \big)^{k-2}  }.
\ealigns
We use $\sigma_{t+2}^2  = a_0 + \sigma_{t+1}^2(a_1 Z_{t+1}^2+b)$ from \eqref{eq:garch_loss} and observe that 
\baligns
&\sum_{k=2}^{m} \sqrt{ a_0 \sum_{j=0}^{k-3} \big( a_1 F_{Z^2}^{-1}(\alpha) + b \big)^j + \Big( a_0 + \sigma_{t+1}^2(a_1 Z_{t+1}^2+b) \Big) \big( a_1 F_{Z^2}^{-1}(\alpha) + b \big)^{k-2}  }\\
&\qquad = \sum_{k=2}^{m}  \sqrt{ a_0 \sum_{j=0}^{k-2} \big( a_1 F_{Z^2}^{-1}(\alpha) + b \big)^j + \Big(\sigma_{t+1}^2(a_1 Z_{t+1}^2+b) \Big) \big( a_1 F_{Z^2}^{-1}(\alpha) + b \big)^{k-2}  }
\ealigns
is a strictly increasing function in $Z^2_{t+1}$. Hence we can proceed by the same argument as in the pretext leading to \eqref{hack001} to achieve ultimately 
\baligns
\var{t}\Big(L_{t+1} &+\sum_{k=2}^{m}  \tildevar{t+1}\Big(L_{t+k}  \Big)\Big) \\
&= \sigma_{t+1} F_Z^{-1}(\alpha) 
 + F_Z^{-1}(\alpha) \sum_{k=2}^{m} \sqrt{ a_0 \sum_{j=0}^{k-2} \big( a_1 F_{Z^2}^{-1}(\alpha) + b \big)^j + \sigma_{t+1}^2\big( a_1 F_{Z^2}^{-1}(\alpha) + b \big)^{k-1}  }\\
& = 
\sigma_{t+1} F_Z^{-1}(\alpha)  +  F_Z^{-1}(\alpha)  \sum_{k=2}^{m} \sqrt{ \cP_{t}^{t+k} \big( a_1 F_{Z^2}^{-1}(\alpha) + b \big)} \\
&=  F_Z^{-1}(\alpha)  \sum_{k=1}^{m} \sqrt{ \cP_{t}^{t+k} \big( a_1 F_{Z^2}^{-1}(\alpha) + b \big)}\\
&=  \sum_{k=1}^{m} \tildevar{t}\big( L_{t+k}\big).
\ealigns
This finishes the proof.
\end{proof}

\section{Conditional time-consistent Average Value-at-Risk}\label{sec:avar}

This section is devoted to the study of time-consistent alternatives for the Average Value-at-Risk (AVaR). 
Due to coherence AVaR is commonly considered as a more reasonable rectification of VaR. 
For more details we refer to \citet[Chapter 4]{foellmerschied}. 
For $t\in\{0,\ldots,T\}$ we define $L^1(\cF_t)$ as the set of all $\IP_t$-integrable losses, where $\IP_t$ denotes the conditional probability with respect to $\cF_t$ and $\IE_t$ the corresponding conditional expectation.
The following definition relates AVaR to VaR.

\begin{definition}\label{def:avar}
Given a loss position $L \in L^1(\cF_T)$ the {\em Average Value-at-Risk at level} $\alpha \in (0,1)$ {\em at time} $t \in \{0,\ldots,T \}$ is  given by
\balign\label{eq:avar}
\avar{t}(L) = \frac{1}{1-\alpha}\int_\alpha^1 \varu{t}(L) \udu. 
\ealign
with $\var{t}(L)$ as in Definition~\ref{def:var}.
\end{definition}

Whereas VaR quantifies the risk associated to one particular level of risk, reflected in the choice of $\alpha$, AVaR as an integrated VaR takes into account VaR at the entire bandwidth of risk levels between $\alpha$ and $1$ and thus better reflects volume of extreme risks that VaR might neglect.

\smallskip

The following is the analog of a fact well-known for unconditional AVaR (e.g. Lemma~2.16 of \citet{mfe}).

\begin{remark}\label{rem:cts}\rm
If the loss position $L\in L^1(\cF_T)$ has a continuous distribution function, then
\balign\label{eq:avar_cts}
\avar{t}(L) = \IE_t \big[ L  ~|~ L > \var{t}(L)\big], \quad t=0,\ldots,T.
\ealign 
 Due to  \eqref{eq:avar_cts}, AVaR is often also referred to as \emph{conditional VaR} or \emph{Expected Shortfall}.
\halmos
\end{remark}

{\bf Assumption B: } \quad
Additionally to Assumptions A we require from now on also that $Z$ has a continuous distribution function.

\subsection{Time-consistent AVaR for single day losses}

We focus again on the GARCH$(1,1)$ model from \eqref{eq:garch_loss}.

\begin{example}\label{ex:1_day_avar}\rm
Assume the setting as in Example~\ref{ex:1_day}. 
For $t=0,\ldots,T-1$ let $L_{t+1}$ be given by \eqref{eq:garch_loss}.
Then  $\avar{t}(L_{t+1})$  is the \emph{$1$-day-ahead-AVaR}.
If the innovations $(Z_t)_{t\geq 0}$ have a continuous distribution function $F_Z$, then
by linearity of the conditional expectation and $\cF_t$-measurability of $\sigma_{t+1}$ we get for $t=0,\ldots,T-1$,
\baligns
\avar{t}(L_{t+1}) &= \sigma_{t+1} \, \IE_t \, \big[ Z_{t+1} ~|~ Z_{t+1} > F_Z^{-1}(\alpha) \big]\\
&= \sigma_{t+1} \, \frac{1}{1-\al} \int_{F_Z^{-1}(\alpha) }^\infty y \ud F_Z(y)\\
&= \sigma_{t+1} \, \avar{}(Z).
\ealigns
This calculation can also be found in \citet{mcneil_frey}.
\halmos
\end{example}

In analogy to Section~\ref{sec:var}, a time-consistent version of AVaR is constructed as follows.

\begin{definition}\label{def::avar_DPP}
Given a loss position $L \in L^1(\cF_T)$  and the $\avar{t}(L)$ as in Definition~\ref{def:avar}. 
Then the  {\em time-consistent  Average Value-at-Risk at level} $\alpha \in (0,1)$ for $L$
 is defined by
\balign\label{eq:avar_DPP}
\tildeavar{T}(L) := L, \quad \tildeavar{t}(L):=\avar{t}\big( \tildeavar{t+1}(L) \big), \quad t=0,\ldots,T-1.
\ealign 
\end{definition}

For the Average Value-at-Risk of the squared loss $L_T^2$ at time $T$ we can derive an explicit formula similar to \eqref{eq:garch_var}. Note that, though AVaR of $L_{T}^2$ allows for an interpretation as the conditonal volatility at time $T$, our purpose of investigation is to employ AVaR of $L_{T}^2$ to derive pragmatic bounds to AVaR itself, see Section~\ref{sec:as_bounds} below.

\medskip

We start with a result analog to Theorem~\ref{prop:garch_var_terminal}, and recall that $\sum_{k=0}^{-1} a_k= 0$.
The proof is given in Appendix~\ref{AppB}.

\begin{theorem}\label{prop:consistent_avar}
Let $(L_t)_{t=0}^T$ be given by the \GARCH(1,1) model \eqref{eq:garch_loss}. 
Then we have for the squared loss $L^2_T\in L^1(\cF_T)$ at terminal time $T$
\balign\label{eq:garch_avar}
\tildeavar{t}(L_T^2) = \frac1{1-\alpha}\int_\alpha^1 F_{Z^2}^{-1}(u) \udu ~ \cP^T_{t}\Big(a_1 \frac{ 1}{1-\alpha}  \int_\alpha^1 F_{Z^2}^{-1}(u) \udu + b \Big), \quad t=0,\ldots,T-1,
\ealign
where $\cP_t^T: \IR \to \IR$ is an $\cF_t$-measurable mapping given by
\baligns
\cP^T_t(x) = a_0 \sum_{k=0}^{T-t-2} x^k + \sigma_{t+1}^2 x^{T-t-1}, \quad t=0,\ldots,T-1.
\ealigns
\end{theorem}

\medskip

It is also possible to derive expressions for $m-$day ahead $\tildeavar{}$. 
As usual we define $\prod_{j=1}^{0} a_j=1$.

\begin{proposition}\label{prop:garch_avar}
Let $(L_t)_{t=0}^T$ be given by the \GARCH(1,1) model \eqref{eq:garch_loss}. 
For $t>0$ define  $\cQ_t^{t+1}=\sigma_{t+1}$ and for fixed $m\ge 2$ 
\balign\label{eq:Q}
\cQ_t^{t+m}(z_1,z_2,\ldots,z_{m-1}) :=   a_0 \sum_{k=0}^{m-2} \prod_{j=1}^{k}  (a_1 z_j+b) + \sigma_{t+1}^2 \prod_{j=1}^{m-1} (a_1z_k+b).
\ealign
Then $\tildeavar{t}(L_{t+1})=\avar{}(Z)\,\sigma_{t+1}$ and for fixed $m\ge 2$
\balign\label{eq:garch_avar2}
\tildeavar{t}(L_{t+m}) &= \frac{ \avar{}(Z)}{(1-\alpha)^{m-1} }  \underbrace{\int_{F_Z^{-1}(\alpha)}^\infty \cdots \int_{F_Z^{-1}(\alpha)}^\infty}_{(m-1)\text{-times}} \sqrt{   \cQ_t^{t+m}(z_1,\ldots,z_{m-1}) }  {\ud F_{Z^2}(z_1) \cdots \ud F_{Z^2}(z_{m-1}).}
\ealign
\end{proposition}

\begin{proof}
For $m=1$ note that by \eqref{eq:avar} $\tildeavar{t}(L_{t+1})=\avar{t}(L_{t+1})= \sigma_{t+1}\frac1{1-\alpha} \int_{F_Z^{-1}(\alpha)}^\infty y \ud (y)$ as in Example~\ref{ex:1_day_avar}. 
For simplicity we set $\ov \kappa = \frac1{1-\alpha} \int_{F_Z^{-1}(\alpha)}^\infty y \ud (y)$.
 For $\tildeavar{t}(L_{t+2})$ we use this and then Lemma~\ref{lemma:var_lemma} and compute for $m=2$
\balign
\tildeavar{t}(L_{t+2}) &= \avar{t} \big(\tildeavar{t+1} (L_{t+2})  \big) = \avar{t} \big(\avar{t+1} (L_{t+2})  \big) \\
&= \avar{t} \big(\ov\kappa \, \sigma_{t+2} \big) \label{1}\\
&=  \ov\kappa  \avar{t} \big( \sigma_{t+2} \big)\label{2}\\
&=  \ov\kappa \, \IE_{t} \big[  \sigma_{t+2}  \mid \sigma_{t+2} > \var{t+1}(\sigma_{t+2}) \big]\nonumber \\
&= \ov\kappa \, \IE_{t}  \big[  \sigma_{t+2}  \mid Z^2_{t+1} > F_{Z^2}^{-1}(\alpha) \big]\nonumber\\
&= \frac{ \ov\kappa }{1-\alpha }  \, \IE_{t}  \big[ \sqrt{  a_0 + \sigma_{t+1}^2 \big(  a_1 Z_{t+1}^2+b \big)  } \1_{\{ Z^2_{t+1} > F_{Z^2}^{-1}(\alpha) \}} \big]\nonumber\\
&= \frac{ \ov\kappa }{1-\alpha }  \, \int_{F_Z^{-1}(\alpha)}^\infty \sqrt{  a_0 + \sigma_{t+1}^2 \big(  a_1 z+b \big)  }{ \ud F_{Z^2}(z),}\nonumber
\ealign
since $\sigma_{t+1}$ is $\cF_t$-measurable.
Assume that  \eqref{eq:garch_avar2} holds for $L_{t+2},\ldots,L_{t+m-1}$.
Then 
\baligns
&\tildeavar{t}(L_{t+m}) = \avar{t} \big( \tildeavar{t+1} \big( (L_{t+m}) \big) \big)\\
&=  \avar{t}  \Big(\frac{ \ov\kappa}{(1-\al)^{m-2}} 
\underbrace{\int_{F_Z^{-1}(\alpha)}^\infty \cdots \int_{F_Z^{-1}(\alpha)}^\infty}_{(m-2)\text{-times}} \sqrt{   \cQ_{t+1}^{(t+1)+(m-1)}(z_1,\ldots,z_{m-2}) }  \ud F_{Z^2}(z_1) \cdots \ud F_{Z^2}(z_{m-2})
\Big),
\ealigns
where
\[
\cQ_{t+1}^{(t+1)+(m-1)}(z_1,z_2,\ldots,z_{m-2}) :=   a_0 \sum_{k=0}^{m-3} \prod_{j=1}^{k}  (a_1 z_j+b) + \sigma_{t+2}^2 \prod_{j=1}^{m-2} (a_1z_k+b).
\]
Setting $\sigma_{t+2}^2 = a_0+\sigma^2_{t+1}(a_1 Z^2_{t+1} + b)$, then since $\sigma^2_{t+1}$ is $\cF_t$-measurable, factorization of $Z^2_{t+1}$ gives another integral and   another factor $1-\al$ in the denominator. 
\end{proof}

\subsection{Almost sure bounds for AVaR}\label{sec:as_bounds}

Finding an analytical expression for $\tildeavar{t}(L_T)$ for the (unsquared) \GARCH(1,1) loss is not straightforward. It is however possible to derive closed form bounds for $\tildeavar{t}(L_T)$.

\subsubsection{AVaR-bounds for single day losses}

We now derive a closed form upper bound to $\tildeavar{}$ which arises from an application of Jensen's inequality. For the proof of the following Proposition we refer to Appendix~\ref{AppB}.

\begin{proposition}\label{prop:avar_jensen}
Let $L_T$ be the loss position at time $T>0$ given by the \GARCH$(1,1)$ model \eqref{eq:garch_loss}. Then
\balign\label{eq:avar_jensen}
\avarup{t}(L_T) := \frac{1}{1-\alpha} \int_\alpha^1 F_Z^{-1}(y) \ud y  \, \sqrt{ \cP^T_{t}\Big( a_1 \frac{1}{1-\alpha}
 \int_\alpha^1 F_{Z^2}^{-1}(u) \udu + b \Big) }, \quad t=0,\ldots, T-1,
\ealign
where $\cP^T_t(\cdot)$ is given by \eqref{eq:poly}, satisfies 
\[
\tildeavar{t} \leq \avarup{t} \quad t=0,\ldots, T-1.
\]
\end{proposition}

An easy alteration of the proof of the previous result yields a closed form lower bound $\avarlow{}$ for $\tildeavar{}$.

\begin{proposition}\label{prop:avar_lower}
Let $L_T$ be the loss position at time $T>0$ given by the \GARCH$(1,1)$ model \eqref{eq:garch_loss}.  Then
\balign\label{eq:avar_lower}
\avarlow{t}(L_T) := \frac{1}{1-\alpha} \int_\alpha^1 F_Z^{-1}(u) \udu \, \left( \frac{1}{1-\alpha}  \int_{\alpha}^1 \sqrt{a_1 F^{-1}_{Z^2}(y) + b} \, \ud y   \right)^{T-t-1} \sigma_{t+1}, \quad t=0,\ldots T-1,
\ealign
satisfies 
\[
\tildeavar{t}(L_T) \geq \avarlow{t}(L_T) \quad t=0,\ldots, T-1.
\]
\end{proposition} 

\subsubsection{AVaR-bounds for aggregated losses}

Unfortunately, for AVaR there exists no result corresponding to Proposition~\ref{prop:var_linear}, hence AVaR does not linearize across aggregation of GARCH losses. 
A key obstacle is that Lemma~\ref{lemma:var_lemma} does not apply. 
However, due to the subadditivity of AVaR and the property \eqref{eq:linearize}, we can derive an upper bound for the aggregation of GARCH losses.

\begin{proposition}\label{prop4.9}
Let $(L_t)_{t=0}^T$ be given by the \GARCH$(1,1)$ model \eqref{eq:garch_loss}.
Then,  for $t\in\{1,\ldots,T-1\}$ and $m\in\IN$ such that $t+m\le T$,
 the $m$-day-ahead $\tildeavar{t}$ of aggregated losses $\sum_{k=1}^m L_{t+k}$ is bounded by
\balign\label{eq:avar_bounds} 
\tildeavar{t} \big( \sum_{k=1}^m L_{t+m+1} \big) \leq \sum_{k=1}^m \avarup{t}(L_{t+m+1}). 
\ealign
\end{proposition}

\begin{proof}
By the subadditivity property \eqref{eq:linearize} $\tildeavar{t} \big( \sum_{k=1}^m L_{t+m+1} \big)$ satisfies
\[
\tildeavar{t} \big( \sum_{k=1}^m L_{t+m+1} \big) \leq \sum_{k=1}^m \tildeavar{t} \big(  L_{t+m+1} \big).
\]
Now the assertion follows from an application of Proposition~\ref{prop:avar_jensen}.
\end{proof}

\begin{remark}\label{rem:avar_lower_approx}\rm
Due to the lack of linearization across aggregation of GARCH losses, the aggregation of the $m$-ahead AVaR bounds from Proposition~\ref{prop:avar_lower} do not produce a proper lower bound for $\tildeavar{t} \big( \sum_{k=1}^m L_{t+m+1} \big)$. 
Whereas in case of the single $m$-day-ahead $\avarlow{t}(L_T)$ is a true lower bound to $\tildeavar{t}(L_T)$,  their aggregation $\sum_{k=1}^m \avarlow{t}(L_{t+m+1})$ is rather a lower bound to the upper bound $\sum_{k=1}^m \avarup{t}(L_{t+m+1})$.
 It can happen that $\sum_{k=1}^m \avarlow{t}(L_{t+m+1})$  is either an upper bound or a lower bound for $\sum_{k=1}^m \tildeavar{t}(L_{t+m+1})$. 
 Nevertheless, we will use $\sum_{k=1}^m \avarlow{t}(L_{t+m+1})$ as a weak lower bound in our numerical experiments to get an orientation about how much the upper bound $\sum_{k=1}^m \avarup{t}(L_{t+m+1})$ is tailing off.
\end{remark}

\section{Extreme value theory based quantile estimation}\label{sec:evt}

\subsection{Generalized Pareto Distribution}\label{sec:GPD}

\begin{figure}[ht!]
	\centering
  \includegraphics[width=16.5cm,height=9cm]{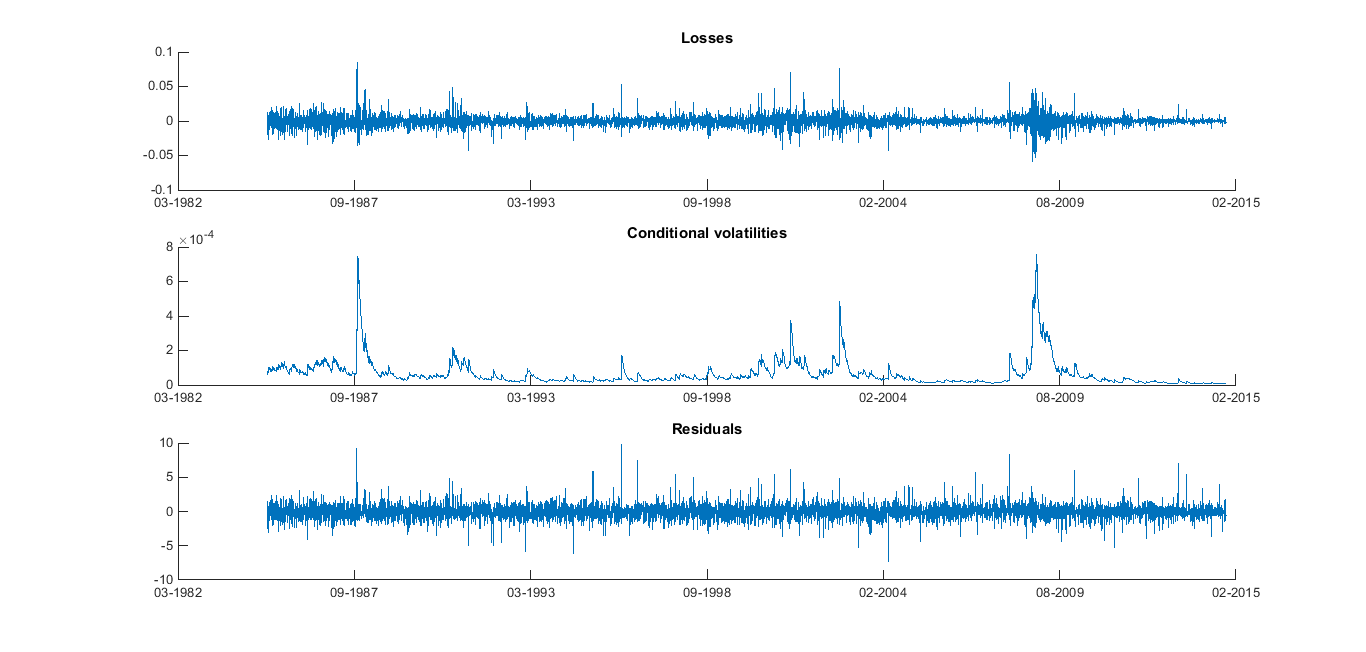}
	\caption{Motorola stock price analysis: loss data (top), conditional variances after fitting a \GARCH$(1,1)$ model (middle), and residuals of loss data (bottom).}
	\label{fig:motorola_returns_vola_resids}
\end{figure}    

Up to now we have not fixed the noise distribution, only assumed certain properties like infinite right endpoint or continuity of the distribution function.
Throughout we worked with $\alpha$ close to 1 corresponding to the noise distribution function to be close to 1.
Thus it is sufficient to specify the distribution function above some high threshold $u$. 
This is a  typical assumption in extreme value theory, and we will apply the Peaks-over-Threshold method (as in \cite{mcneil_frey}).
We first explain the setting in general.

\medskip
 
The \emph{Generalized Pareto distribution (GPD)} is given by
\balign\label{eq:gpd}
G_{\xi,\beta}(x) = \begin{cases} 1- \big( 1+ \dfrac{\xi}{\beta} x \big)^{-1/\xi}, &\xi \neq 0,\\  
1 - \exp\big( - \dfrac{x}{\beta} \big),	& \xi=0,
\end{cases}
\ealign
where $\beta>0$ and $\xi \in \IR$. If $\xi>0$ \eqref{eq:gpd} is defined for $x\geq 0$ and if $\xi<0$ \eqref{eq:gpd} is defined on $x \in [0, - \beta/\xi]$, see e.g. Section~3.4 in \citet{ekm}. Assume that we fix some high threshold $u>0$. Given a random variable $X$ with distribution function $F$ and right endpoint $x_F$, its associated \emph{excess distribution function} is defined as
\balign\label{eq:ex_dist}
F_u(y) = \IP\Big( X-u \leq y~|~X>u\Big) = \frac{F(y+u)-F(u)}{1-F(u)}, \quad 0 \leq y < x_F -u.
\ealign
The strength of the GPD is compressed in a result by \citet{pick} and \citet{BdH} which classifies the GPD as the limit distribution of a large class of excess distributions. More precisely, under mild conditions there exists a measurable non-negative parameter $\beta=\beta(u)$ such that 
\[
\lim_{u\to x_F} \sup_{0 \leq x \leq x_F - u} |F_u(x) - G_{\xi,\beta(u)}| = 0
\]
holds, see Theorem 3.4.13 in \citet{ekm} for a rigorous statement of this result. 
The density of \eqref{eq:gpd} is given by
\balign\label{eq:gpd_density}
g_{\xi,\beta}(x) = \begin{cases} \dfrac1\beta \big( 1+ \dfrac{\xi}{\beta} x \big)^{-1/\xi-1}, &\xi \neq 0,\\  
\dfrac1\beta\exp\big( -\dfrac{ x}{\beta} \big),	& \xi=0.
\end{cases}
\ealign
Under the assumption that $Z$ has the distribution function $F_Z$, which for some high enough threshold $u>0$ satisfies $F_u(x) = G_{\xi,\beta}(x)$ for $0 \leq x \leq x_F - u$ and for some $\xi \in \IR$ and $\beta>0$,  we find for $\alpha \geq F(u)$ (for $\xi=0$ we interpret this quantile as the quantile of the corresponding exponential distribution)
\balign
F_Z^{-1}(\alpha) &= u + \frac\beta\xi\Big( \Big( \frac{1-\alpha}{1-F(u)} \Big)^{{-\xi}}  -1\Big),\label{eq:var_02}\\
\frac{1}{1-\alpha} \int_\alpha^1 F_Z^{-1}(y) \ud y &= \frac{F_Z^{-1}(\alpha)}{1-\xi} + \frac{\beta-\xi u}{1-\xi}.\label{eq:avar_02}
\ealign   
By \eqref{comp} we obtain
$$F_{Z^2}^{-1}(\alpha) = F_Z^{-1}(\frac12(\alpha+1))^2 = \Big(
 u + \frac\beta\xi\Big( \Big( \frac{\frac12(1-\alpha)}{1-F(u)} \Big)^{-\xi}  -1\Big)
\Big)^2.$$
Unfortunately, there is no explicit expression for $\frac1{\alpha-1}\int_\alpha^1 F_{Z^2}^{-1}(y) \ud y$.

\subsection{Statistical model fitting}\label{sec:stats}

\begin{figure}[t!]
\hspace*{-1cm}  \includegraphics[scale=0.52]{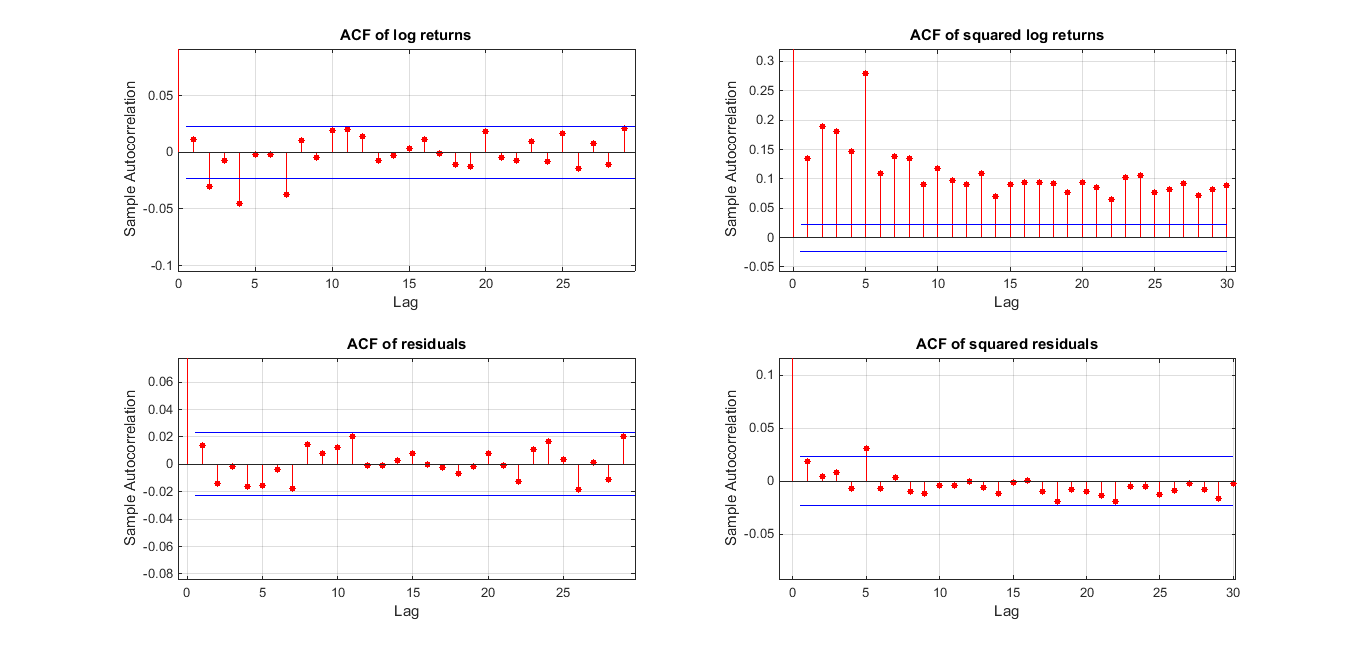}
	\caption{Motorola stock price analysis: sample autocorrelation functions for loss data (top) and residuals after fitting a GARCH$(1,1)$ model (bottom).}
	\label{fig:motorola_acfs}
\end{figure}

In this section we apply the theory and formulas derived previously to a data set. We choose the historical daily closing prices of the Motorola stock from 1st March 1985 until 15th October 2014 as this data set provides several canonical features of financial time series. We transform prices into negative log-returns; i.e., into losses, and fit the \GARCH$(1,1)$ parameters using  Quasi Maximum Likelihood Estimation (QMLE)  (e.g. \cite{FZ}, Chapter~7). The parameter estimates can be found in Table \ref{table:garch_parameters}, and
the outcome is depicted in Figure~\ref{fig:motorola_returns_vola_resids}. 

\medskip 

\begin{figure}[t!]
	\centering
   \includegraphics[width=7cm, height=5.5cm]{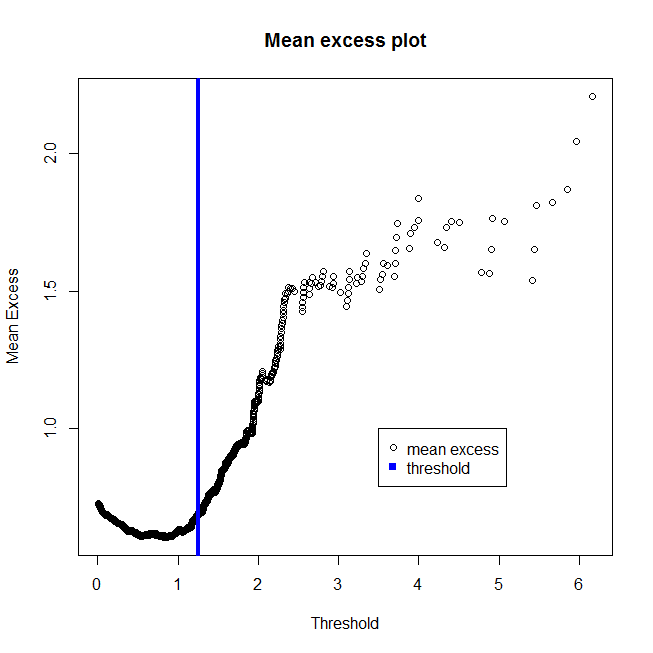} \hspace*{0.1cm}
    \includegraphics[width=7cm, height=5.5cm]{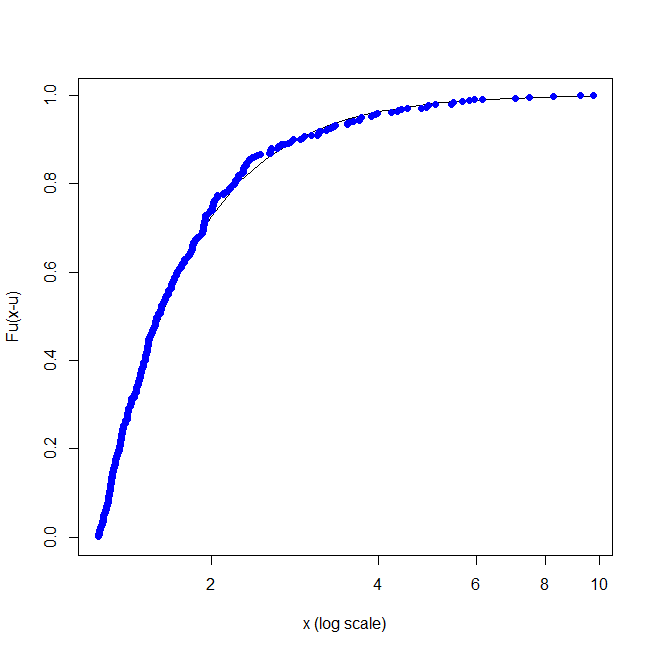}
	\caption{Fit of the Generalized Pareto Distribution. Left: mean excess plot of the positive residuals, and QQ-plot of the threshold exceeding residuals against the fitted GPD. Right: excess distribution $F_u(x-u)$ from the fitted GPD model (solid line) against the empirical estimates of excess probabilities (dotted points).}
	\label{fig:motorola_qqplot_expo_gpd}
\end{figure}

\begin{table}[h!]                      
\centering                             
\begin{tabular}{c|ccc}                                                 
Parameter&Value&Standard error&\\
\hline
$\wh a_0$&2e-07&1.09e-07&\\
$\wh a_1$&0.0451&0.0014&\\
$\wh{b}$&0.9531&0.0013&\\
 \hline
\end{tabular}                                 
\caption{Estimated GARCH$(1,1)$ parameters by QMLE.}               
\label{table:garch_parameters}             
\end{table}

We see in the middle plot of Figure~\ref{fig:motorola_returns_vola_resids} major clustering of volatility in October 1987 (Black Monday), in a pronounced period between 2000 until 2002 (Dot-com bubble and wake of 9/11 attacks) and in a longer lasting period following the financial crisis between 2008 until 2010. 

In a next step we examine the sample autocorrelation functions of the loss data and the residuals after fitting a \GARCH$(1,1)$ model. 
In Figure~\ref{fig:motorola_acfs} the bottom plots depict the acf of the residuals  and the squared residuals and is supportive for the our assumption of i.i.d. GARCH residuals $Z_t$. 
This is also reflected in several runs of the Ljung-Box for various lags for the residuals. The residuals also pass the augmented Dickey-Fuller and the KPSS stationarity tests.

\begin{figure}[b!]
	\centering
  \includegraphics[height=5cm]{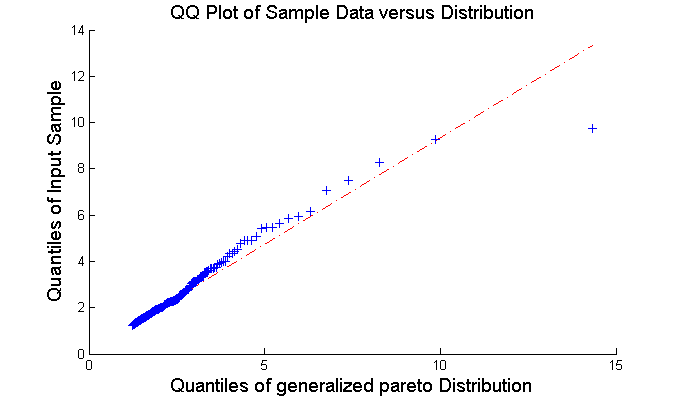}
	\caption{QQ-plot of the threshold exceeding residuals against the fitted GPD.}
	\label{fig:motorola_qqplot_expo_gpd2}
\end{figure}

\medskip

As explained in Section~\ref{sec:GPD} we fit a GPD to the upper tail of the residuals. 
We first have to choose a high enough threshold value $u$ and we choose it as the approximate $92\%$ quantile of the residuals.
This is supported by studying the mean excess plot of the nonnegative residuals in Figure~\ref{fig:motorola_qqplot_expo_gpd}: the $92\%$ quantile of the residuals (solid blue line) yields a threshold which sufficiently marks the beginning of the linear behaviour of the mean excess plot.
Since the empirical mean excesses are increasing, we may assume that the shape parameter $\xi$ is positive.
This is confirmed by the parameter estimates for $\xi$ and $\beta$.
The Maximum Likelihood Estimators are $\wh\xi= 0.3376$ with a $95\%$ confidence interval $[0.2272, 0.4481]$ and $\wh\beta = 0.4609$ with a $95\%$ confidence interval $[0.4023, 0.5280]$.

\medskip 

In Figure~\ref{fig:motorola_qqplot_expo_gpd}, the right hand plot depicts the GPD fit of the excess distribution $F_u(x-u) = \IP(X\leq x\mid X>u)$ superimposed on empirical estimates of excess probabilities. Note how well the GPD model fits to the empirical estimates of the excess probabilities.

\medskip	

A QQ-plot of the empirical quantiles against the fitted quantiles is depicted in Figure~\ref{fig:motorola_qqplot_expo_gpd2}. 
Note again the good correspondence of the fitted GPD with the empirical estimates.

\subsection{Fitting time-consistent risk measures to data}\label{sec:numerics}

We now compute the corresponding time-consistent risk measures from Sections~\ref{sec:var} and~\ref{sec:avar}. 

\subsubsection{Time consistent VaR estimation}

In a first step, for a single loss position $L_t$ we compute the $m$-day-ahead time-consistent VaR given by Proposition~\ref{prop:garch_var_terminal} for different levels of $\alpha$; i.e., we fix $t$ and consider $\tildevar{t}(L_{t+m})$ for various $m\in\IN$. 
\begin{figure}[b]
	\centering
  \includegraphics[width=6.5cm,height=5.6cm]{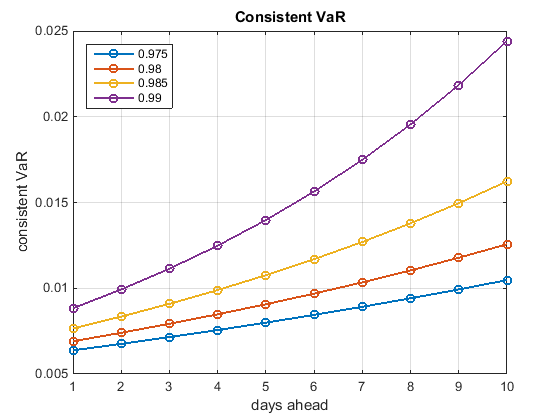}\hspace{0.001cm} \includegraphics[width=9.9cm,height=5.9cm]{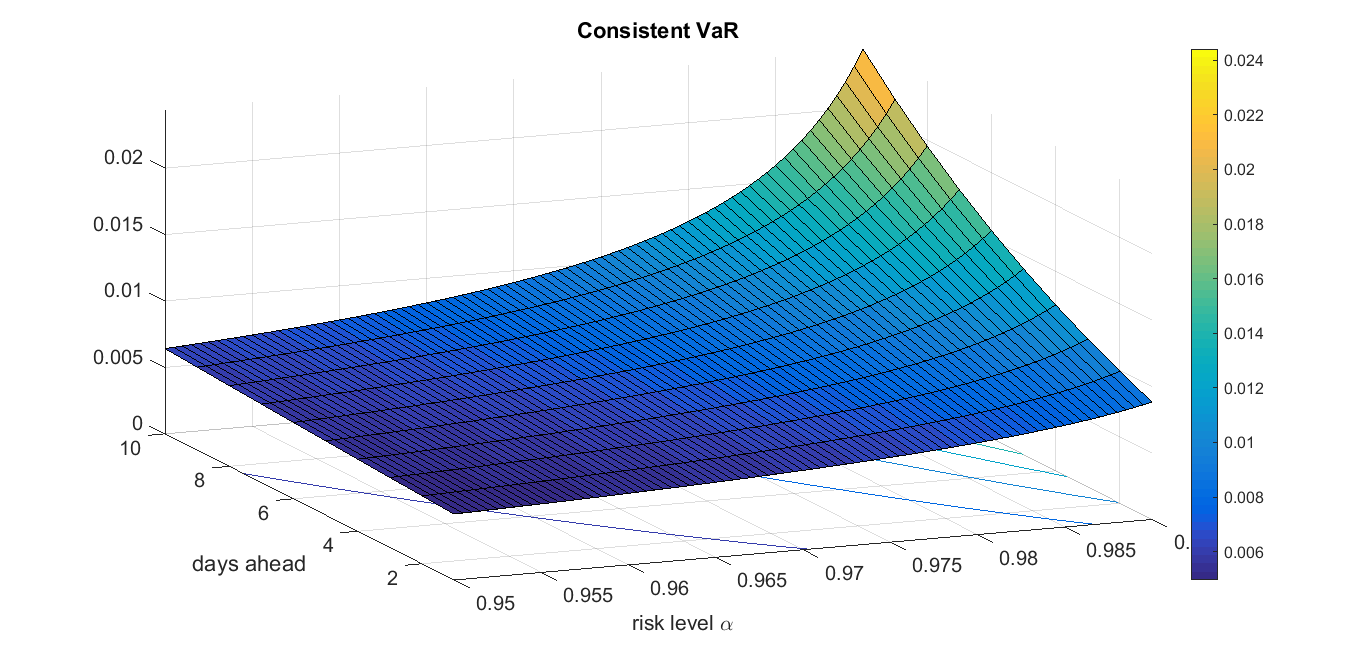}
	\caption{Time-consistent VaR estimation for single loss positions $m$ days ahead. }
	\label{fig:motorola_var_consistent}
\end{figure}
In Figure~\ref{fig:motorola_var_consistent} we plot $\tildevar{t}(L_{t+m})$ for $m=1,\ldots,10$. 

\medskip

Once the single time consistent risk measures $\tildevar{t}(L_{t+m})$ are computed, we simultaneously get the risk measure of the aggregated losses over $m$ days from Proposition~\ref{prop:var_linear} by aggregation; i.e., 
\[
\tildevar{t}\big( \sum_{j=1}^m L_{t+j} \big) = \sum_{j=1}^m \tildevar{t}(  L_{t+j} ).
\]
\begin{figure}
	\centering
  \includegraphics[width=6.5cm,height=5.6cm]{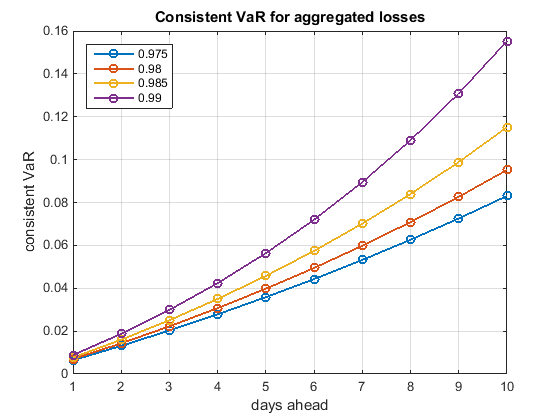}\hspace{0.001cm} \includegraphics[width=9.9cm,height=5.6cm]{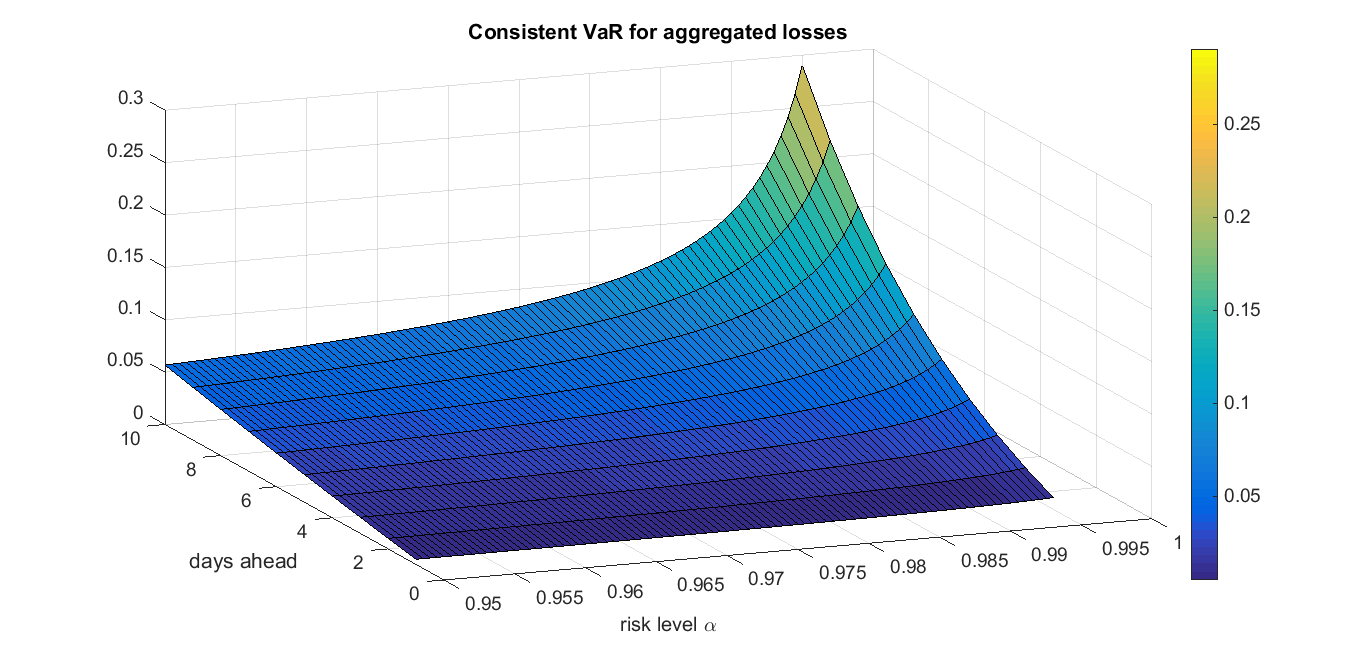}
	\caption{Time-consistent VaR estimation for aggregated loss positions $m$ days ahead.}
	\label{fig:motorola_var_consistent_aggr}
\end{figure}
In Figure~\ref{fig:motorola_var_consistent_aggr} we plot $\tildevar{t}(\sum_{j=1}^m L_{t+j})$ for $m=1,\ldots,10$. 
Table~\ref{table:var_single_aggr} shows the values of $\tildevar{}$ for single losses $L_{t+m}$ and aggregated losses $\sum_{j=1}^m L_{t+j}$ for various levels of $\alpha$ and $m=1,\ldots,10$. 

\begin{table}[htb!]                      
\centering                             
\begin{tabular}{c|cccc||cccc}                                                 
&&single loss&&&&aggr. loss&&\\ 
\hline
$\alpha/m$ & 97.5\% & 98\% & 98.5\% & 99\% & 97.5\% & 98\% & 98.5\% & 99\% \\        
 \hline
1 & 0.0064 & 0.0069 & 0.0077 & 0.0088 & 0.0064 & 0.0069 & 0.0077 & 0.0088 \\
2 & 0.0068 & 0.0074 & 0.0083 & 0.0099 & 0.0132 & 0.0143 & 0.0160 & 0.0187 \\
3 & 0.0072 & 0.0079 & 0.0091 & 0.0111 & 0.0204 & 0.0222 & 0.0251 & 0.0297 \\
4 & 0.0076 & 0.0085 & 0.0099 & 0.0125 & 0.0280 & 0.0307 & 0.0350 & 0.0422 \\
5 & 0.0080 & 0.0091 & 0.0108 & 0.0140 & 0.0360 & 0.0398 & 0.0458 & 0.0562 \\
6 & 0.0084 & 0.0097 & 0.0117 & 0.0156 & 0.0444 & 0.0495 & 0.0575 & 0.0718 \\
7 & 0.0089 & 0.0103 & 0.0127 & 0.0175 & 0.0533 & 0.0598 & 0.0702 & 0.0893 \\
8 & 0.0094 & 0.0110 & 0.0138 & 0.0196 & 0.0627 & 0.0708 & 0.0840 & 0.1089 \\
9 & 0.0099 & 0.0118 & 0.0150 & 0.0218 & 0.0726 & 0.0826 & 0.0990 & 0.1307 \\
10 & 0.0105 & 0.0126 & 0.0162 & 0.0244 & 0.0831 & 0.0952 & 0.1152 & 0.1551 \\
\hline
\end{tabular}                                 
\caption{Values for $\tildevar{}$ for single losses $L_{t+m}$ and aggregated losses $\sum_{j=1}^m L_{t+j}$. }               
\label{table:var_single_aggr}             
\end{table}

\subsubsection{Time consistent AVaR estimation}

In a second step, we compute the approximate upper and lower AVaR bounds for single  loss position $L_t$ and we compute the $m$-day-ahead for different levels of $\alpha$; i.e., we fix $t$ and consider $\avarlow{t}\big( L_{t+m} \big)$ and $\avarup{t}\big( L_{t+m} \big)$ for various $m\in\IN$.
The risk measure of the aggregated losses over $m$ days we obtain from Propositions~\ref{prop:avar_jensen}.

\begin{table}[htb!]                                                                                              
\centering                                                                                                     
\begin{tabular}{c|cccc||cccc}                                                 
&&lower&bound &&&upper&bound&\\
\hline
$\alpha/m$ & 97.5\% & 98\% & 98.5\% & 99\% & 97.5\% & 98\% & 98.5\% & 99\% \\      
 \hline  
1 & 0.0098 & 0.0106 & 0.0118 & 0.0135 & 0.0098 & 0.0106 & 0.0118 & 0.0135 \\
2 & 0.0115 & 0.0127 & 0.0145 & 0.0177 & 0.0131 & 0.0146 & 0.0168 & 0.0206 \\
3 & 0.0134 & 0.0152 & 0.0180 & 0.0231 & 0.0173 & 0.0199 & 0.0239 & 0.0313 \\
4 & 0.0157 & 0.0182 & 0.0222 & 0.0302 & 0.0230 & 0.0271 & 0.0340 & 0.0475 \\
5 & 0.0183 & 0.0217 & 0.0275 & 0.0394 & 0.0304 & 0.0370 & 0.0483 & 0.0721 \\
6 & 0.0214 & 0.0260 & 0.0340 & 0.0515 & 0.0402 & 0.0504 & 0.0686 & 0.1093 \\
7 & 0.0250 & 0.0311 & 0.0420 & 0.0672 & 0.0532 & 0.0687 & 0.0974 & 0.1658 \\
8 & 0.0291 & 0.0371 & 0.0520 & 0.0878 & 0.0704 & 0.0936 & 0.1384 & 0.2515 \\
9 & 0.0340 & 0.0444 & 0.0643 & 0.1146 & 0.0932 & 0.1276 & 0.1966 & 0.3813 \\
10 & 0.0398 & 0.0531 & 0.0795 & 0.1497 & 0.1232 & 0.1738 & 0.2792 & 0.5783 \\
\hline
\end{tabular}                                                                                                                   
\caption{Values for $\avarlow{}$ and $\avarup{}$ for single losses $L_{t+m}$ with $m=1,\ldots,10$.} 
\label{table:val_low_jensen}                                                                                     
\end{table} 

\bigskip

\begin{table}[ht]                                                          
\centering                                                                 
\begin{tabular}{c|cccc||cccc}                                                 
&&aggr. $\avarlow{}$&&&&aggr. $\avarup{}$&&\\
&&&&&&&&\\
$\alpha/m$ & 97.5\% & 98\% & 98.5\% & 99\% & 97.5\% & 98\% & 98.5\% & 99\% \\        
 \hline
1 & 0.0098 & 0.0106 & 0.0118 & 0.0135 & 0.0098 & 0.0106 & 0.0118 & 0.0135 \\
2 & 0.0213 & 0.0233 & 0.0263 & 0.0312 & 0.0229 & 0.0252 & 0.0286 & 0.0341 \\
3 & 0.0347 & 0.0385 & 0.0443 & 0.0543 & 0.0402 & 0.0451 & 0.0525 & 0.0654 \\
4 & 0.0504 & 0.0567 & 0.0665 & 0.0845 & 0.0632 & 0.0722 & 0.0865 & 0.1129 \\
5 & 0.0687 & 0.0784 & 0.0940 & 0.1239 & 0.0936 & 0.1092 & 0.1348 & 0.1850 \\
6 & 0.0901 & 0.1044 & 0.1280 & 0.1754 & 0.1338 & 0.1596 & 0.2034 & 0.2943 \\
7 & 0.1151 & 0.1355 & 0.1700 & 0.2426 & 0.1870 & 0.2283 & 0.3008 & 0.4601 \\
8 & 0.1442 & 0.1726 & 0.2220 & 0.3304 & 0.2574 & 0.3219 & 0.4392 & 0.7116 \\
9 & 0.1782 & 0.2170 & 0.2863 & 0.4450 & 0.3506 & 0.4495 & 0.6358 & 1.0929 \\
10 & 0.2180 & 0.2701 & 0.3658 & 0.5947 & 0.4738 & 0.6233 & 0.9150 & 1.6712 \\
\hline
\end{tabular}                                                              
\caption{{$\avarlow{}$ and} $\avarup{}$ for aggregated losses $\sum_{j=1}^m L_{t+j}$.}                                                   
\label{table:avar_aggr}                                                 
\end{table}

\subsection{Conclusions}

Obviously, the interpretation for the dynamic time-consistent (A)VaR differs considerably to that of the static (A)VaR: the dynamic (A)VaR evolves via the composition of the static (A)VaR over time. 
This results in a much more conservative risk measurement as the risky positions that are due far in the future not only enter the risk assessment through their own dynamics at the future maturity but rather enter through their risk assessment along any time point up to maturity. This has the intended effect that risky effects which arise \underline{until maturity} are cushioned at any time. 
As one would expect, the higher safety margins $\alpha$ are required the more dramatic is the increase of safety capital when more days-ahead risk management is envisioned. 

\medskip

Table~\ref{table:var_single_aggr} contrasts single and aggregated time-consistent VaR values for different $\alpha$ and maturities $m$.
It shows convincingly, how much higher capital reserves are needed to guarantee uniform safety at the same level over the whole time to maturity. 
Already at a level of $\alpha=0.975$ the time-consistent aggregate loss VaR more than doubles from maturity 1 to 2 and multiplies by a factor of more than 12 to maturity 10. 
There is a high price to pay to safeguard against all uncertainties, which may lie in the far future.

\medskip

For a comparison recall a standard industry method to estimate a 10-day VaR based on the central limit theorem, or normality of future losses (e.g. \citet{mcneilfreyem}, Section~{2.3.4}).
Recall that the loss from time $t$ over the next $m$ periods can be written as the sum over the negative returns during this period. If returns are iid with mean zero and variance $\sigma^2$ (or even normally distributed), then this sum is again (approximately) normally distributed with mean zero and variance $m \sigma^2$. 
This motivates the estimation of the sum of losses over $m$ days by the estimation of the 1-day VaR and multiply it by $\sqrt{m}$.

\medskip

Let us compare the values for $\tildevar{t}$ from Table~\ref{table:var_single_aggr}  with this industry standard.
We find for $\alpha=0.975$ a 10-day VaR of $0.0064\,\sqrt{10} = 0.0202$ (which we have to compare with the time-consistent $\tildevar{t}(\sum_{j=1}^{10} L_{t+j})=0.0830$, which is more than 4 times as large), and for $\alpha=0.99$ a 10-day VaR of $0.0088\,\sqrt{10}=0.0278$ (which we have to compare with the time-consistent $\tildevar{t}(\sum_{j=1}^{10} L_{t+j})=0.1553$, which is more than 5 times as large).  One reason for this huge difference is the well-known fact that GARCH losses do not scale with $\sqrt{m}$, but scaling depends strongly on the parameters; cf. \cite{FZ}, Chapter~4. 
However, this alone does not explain the huge difference between the simple industry standard and the time-consistent VaR for the aggregated losses.

\medskip

Due to their construction the composed VaR and AVaR for aggregated future losses produce much more conservative reserve requirements than the standard VaR and AVaR for the same level of $\alpha$. 
As an implication the standard reserving requirement of excessively high levels of $\alpha$ like $99\%$ or $99.9\%$ covering $100$- or $1000$-year events may be put to a test taking into consideration reduced levels of $\alpha$, e.g. in the bandwidth $90\% - 97.5\%$. 
The reduction of such extremely high levels would also be very reasonable from a statistical point of view as lower level quantiles give rise to much more reliable estimators.


\appendix

\section{Proofs of Section~3}\label{AppA}

\noindent
{\em Proof of Theorem~\ref{prop:garch_var_terminal}  }\quad 
We proceed by backward induction. Firstly, by \eqref{varL}, at $T-1$ we have the 1-day-ahead-VaR 
\[
\tildevar{T-1}(L_T) = F_Z^{-1}(\al)~ \sigma_T 
\] 
which agrees with \eqref{eq:garch_var} for $t=T-1$. Assume that \eqref{eq:garch_var} holds for all $s=t,\ldots,T-1$. We have
\baligns
\tildevar{t-1}(L_T) &= \var{t-1}\big( \tildevar{t}(L_T) \big)\\
&= \var{t-1}\big( F_Z^{-1}(\al) \sqrt{ \cP^T_{t}\big( a_1 F_{Z^2}^{-1}(\alpha) + b \big) } \big)\\
&= \essinf\big\{m \in L^0(\cF_{t-1}): ~ \IP( F_Z^{-1}(\al) \sqrt{ \cP^T_{t}\big( a_1 F_{Z^2}^{-1}(\alpha) + b \big) } \le m  ~|~\cF_{t-1}) \geq \alpha \big\}.
\ealigns
We denote by $\IP_{t-1}$ the conditional probability with respect to $\cF_{t-1}$.
Note that 
\baligns
\lefteqn{\IP_{t-1}( F_Z^{-1}(\al) \sqrt{ \cP^T_{t}\big( a_1 F_{Z^2}^{-1}(\alpha) + b \big) } \leq m) 
= \IP_{t-1}( F_Z^{-1}(\al)^2 \cP^T_{t}\big( a_1 F_{Z^2}^{-1}(\alpha) + b \big) \leq m^2)}\\
&= \IP\Big( \sigma_{t+1}^2 (a_1F_{Z^2}^{-1}(\alpha) +b)^{T-t-1} \leq \Big(\frac{m}{F_Z^{-1}(\al)}\Big)^2 - a_0 \sum_{k=0}^{T-t-2} (a_1F_{Z^2}^{-1}(\alpha) +b)^k  ~|~\cF_{t-1} \Big).
\ealigns
Using the definition of the GARCH volatility \eqref{eq:garch_loss} for $\si_{t+1}^2$ this can be continued by
\baligns
&\IP_{t-1}( F_Z^{-1}(\al) \sqrt{ \cP^T_{t}\big( a_1 F_{Z^2}^{-1}(\alpha) + b \big) } \leq m) \\
&\qquad =\IP_{t-1}\Big( a_1 \sigma_t^2 Z_t^2 \leq \frac{1}{(a_1F_{Z^2}^{-1}(\al) +b)^{T-t-1}}\left( \left(\frac{m}{F_Z^{-1}(\al)}\right)^2 - a_0 \sum_{k=0}^{T-t-2} (a_1F_{Z^2}^{-1}(\alpha) +b)^k \right) -a_0 - b\sigma_t^2 \Big) \\
\ealigns
Since $\si_t$ is $\cF_{t-1}$-measurable and $Z_t$ is independent of $\cF_{t-1}$ we conclude that 
\baligns
\tildevar{t-1}(L_T) 
&= F_Z^{-1}(\al) \sqrt{ \big( a_0+ (a_1 F_{Z^2}^{-1}(\alpha)   + b) \sigma_t^2 \big) \big( a_1 F_{Z^2}^{-1}(\alpha) +b \big)^{T-t-1} + a_0 \sum_{k=0}^{T-t-2}(a_1F_{Z^2}^{-1}(\alpha)+b)^k }\\
&= F_Z^{-1}(\al) \sqrt{ a_0 \sum_{k=0}^{T-t-1}(a_1F_{Z^2}^{-1}(\alpha)+b)^k + \sigma_t^2(a_1F_{Z^2}^{-1}(\alpha)+b)^{T-t}}\\
&= F_Z^{-1}(\al) \sqrt{\cP^T_{t-1}(a_1F_{Z^2}^{-1}(\alpha)+b)}.
\ealigns
This finishes the proof.
\halmos

\section{Proofs of Section~4}\label{AppB}

We need the following lemma. 

\begin{lemma}\label{lemma:var_lemma}
For $t=0,\ldots,T-2$ assume that $f_t: \IR \to \IR$ is a $\cF_t$-measurable,  and strictly increasing mapping. Then we have
\[
\{\omega \in \Omega: f_t(Z_{t+1}^2) > \var{t}\big( f_t(Z_{t+1}^2) \big) \} \, = \, \{\omega \in \Omega: Z^2_{t+1} > F_{Z^2}^{-1}(\alpha) \}.
\]
In particular,
$$\{\omega \in \Omega: \sigma_{t+2} > \var{t}(\sigma_{t+2}) \} \, = \, \{\omega \in \Omega: Z^2_{t+1} > F_{Z^2}^{-1}(\alpha) \}.$$
\end{lemma}

\begin{proof}
Due to the assumptions on $f_t$ it is invertible. According to the definition of $\var{t}$ we have
\baligns
\var{t}\big( f_t(Z^2_{t+1}) \big) &= \essinf\big\{m \in L^0(\cF_t): ~ \IP( f_t(Z^2_{t+1}) \le m  ~|~\cF_t) \geq \alpha \big\}\\
&= \essinf\big\{m \in L^0(\cF_t): ~ \IP( Z^2_{t+1} \le f_t^{-1}(m) ) \geq \alpha \big\}\\
&= \essinf\big\{m \in L^0(\cF_t): ~ F_{Z^2}\big( f_t^{-1}(m) \big) \geq \alpha \big\}\\
&= \essinf\big\{m \in L^0(\cF_t): ~ m \geq f_t\big( F_{Z^2}^{-1}(\alpha) \big) \big\}\\
&=  f_t\big( F_{Z^2}^{-1}(\alpha) \big)
\ealigns
where the third line follows from the independence between $Z_{t+1}$ and $\cF_t$. Thus
\[
f_t(Z_{t+1}^2) > \var{t}\big( f_t(Z^2_{t+1}) \big) = f_t\big( F_{Z^2}^{-1}(\alpha) \big)
\]
holds if and only if $Z^2_{t+1} > F_{Z^2}^{-1}(\alpha)$. For the second part, note that by
\eqref{var1},
\baligns
\var{t}(\sigma_{t+2}) &= \sqrt{ a_0 + \sigma_{t+1}^2 \big( a_1 F_{Z^2}^{-1}(\alpha) + b \big)}.
\ealigns
From the definition of the \GARCH(1,1) model \eqref{eq:garch_loss} we conclude
\baligns
\sigma_{t+2} = \sqrt{a_0 + \sigma_{t+1}^2 \big( a_1 Z_{t+1}^2 + b \big) } = f_t(Z^2_{t+1}) >  \var{t}\big( f_t(Z^2_{t+1}) \big)  = \sqrt{a_0 + \sigma_{t+1}^2 \big( a_1 F_{Z^2}^{-1}(\alpha) + b \big) } = \var{t}(\sigma_{t+2})
\ealigns
if and only if  $Z^2_{t+1} > F_{Z^2}^{-1}(\alpha)$. 
\end{proof}

\medskip

\noindent
{\em Proof of Theorem~\ref{prop:consistent_avar}  }\quad
We apply again backward induction. 
From \eqref{eq:avar} and Example~\ref{ex:1_day_avar} we have
$$\tildeavar{T-1}(L_T^2) = \sigma_T^2 \, \frac{1}{1-\alpha}  \int_\alpha^1 F_{Z^2}^{-1}(u)  \udu,$$ 
which agrees with \eqref{eq:garch_avar} for $t=T-1$. 
For simplicity we write $\ov\kappa_2 =  \frac{1}{1-\alpha}  \int_\alpha^1 F_{Z^2}^{-1}(u)  \udu$.
Now assume that \eqref{eq:garch_avar} holds for all $s=t,\ldots,T-1$. Then it remains to prove \eqref{eq:garch_avar} for $t-1$. 
We have by \eqref{eq:avar} and \eqref{eq:garch_loss}
\baligns
\tildeavar{t-1}(L_T^2) &= \avar{t-1}\big( \tildeavar{t}(L_T^2) \big)
= \avar{t-1}\left(  \ov\kappa_2 ~ \cP^T_{t}\Big( a_1\ov\kappa_2+ b \Big) \right).
\ealigns
We denote $G_t:= \cP^T_{t}( a_1\ov\kappa_2 + b )$, which is a measurable function of $\sigma_{t+1}$, and take the constant out of the expectation, which 
yields
\baligns
\tildeavar{t-1}(L_T^2) 
&= \ov\kappa_2 ~  \IE_{t-1} \, \big[ G_t  ~|~ G_t > \var{t-1}(G_t) \big].
\ealigns
Now note that by Definition~\ref{def:var}
\baligns
\var{t-1}(G_t) 
&=  \essinf\{ m \in L^0(\cF_{t-1}): \IP\big( G_t \leq m ~|~ \cF_{t-1} \big) > \alpha \}.
\ealigns
We denote by $\IP_{t-1}$ the conditional probability with respect to $\cF_{t-1}$.
We compute further, using the definition of the GARCH volatility \eqref{eq:garch_loss} for $\si_{t+1}^2$
\baligns
\IP_{t-1}\big( G_t \leq m \big)
&= \IP_{t-1}\Big( a_0 \sum_{k=0}^{T-t-2} \big( a_1\ov\kappa_2 + b \big)^k + \sigma_{t+1}^2 \big( a_1 \ov\kappa_2 + b \big)^{T-t-1} \leq m \Big)\\
&\qquad= \IP_{t-1}\Big( \sigma_{t+1} \leq\frac{m - a_0 \sum_{k=0}^{T-t-2} ( a_1 \ov\kappa_2+ b )^k }{( a_1\ov\kappa_2 + b )^{T-t-1}}\Big)\\
&\qquad= \IP\Big( Z_t^2 \leq  \Big( \frac{m - a_0 \sum_{k=0}^{T-t-2} ( a_1\ov\kappa_2+ b)^k }{(a_1\ov\kappa_2 + b )^{T-t-1}} - a_0 - b\sigma_t \Big) \frac1{ a_1 \sigma_t^2} \Big),
\ealigns
where in the last line we have used that $\si_t$ is $\cF_{t-1}$-measurable  and the independence of $Z_t$ and $\cF_{t-1}$. 
We can thus conclude that
\baligns
\var{t-1}(G_t) &= \Big( a_0 + \sigma_t^2 (a_1  F_{Z^2}^{-1}(\alpha) + b) \Big) \big( a_1\ov\kappa_2 + b \big)^{T-t-1} + a_0 \sum_{k=0}^{T-t-2} \big( a_1\ov\kappa_2 + b \big)^k.
\ealigns
From Lemma~\ref{lemma:var_lemma} we know that $\{ G_t > \var{t}(G_t)\} =\{Z^2_t > F_{Z^2}^{-1}(\alpha)\}$. 
Hence, it follows from the independence of $Z_t$ and $\cF_{t-1}$ that 
\balign\label{eq:hilferuf001}
\IE_{t-1} \, \big[ G_t  ~|~ G_t > \var{t-1}(G_t) \big] &= \IE_{t-1} \big[ G_t  ~|~ Z^2_t > F_{Z^2}^{-1}(\alpha) \big]
= \frac{1}{1-\alpha}~\IE_{t-1} \big[ G_t  \1_{\{ Z^2_t > F_{Z^2}^{-1}(\alpha) \} } \big].
\ealign
Moreover, we calculate 
\baligns
\lefteqn{\IE_{t-1} \big[ G_t  \1_{\{ Z^2_t > F_{Z^2}^{-1}(\alpha) \} } \big]}\\
&= (1-\alpha)  a_0 \sum_{k=0}^{T-t-2} \big(a_1\ov\kappa_2 + b \big)^k  
 + \big( a_1 \ov\kappa_2+ b \big)^{T-t-1} \int_{F_{Z^2}^{-1}(\alpha)}^\infty \big( a_0 +\sigma_t^2( a_1 u + b ) \big) \ud F_{Z^2}(u)\\
&= (1-\alpha) a_0 \sum_{k=0}^{T-t-2} \big(  a_1 \ov\kappa_2 + b \big)^k 
 + \big( a_1 \ov\kappa_2+ b \big)^{T-t-1} \int_\alpha^1 \big( a_0 + \sigma_t^2 (a_1 F_{Z^2}^{-1}(u) + b) \big) \udu\\
&= (1-\alpha) a_0 \sum_{k=0}^{T-t-1} \big( a_1\ov\kappa_2 + b \big)^k 
 + \big(  a_1\ov\kappa_2+ b \big)^{T-t-1}\big(  a_1\ov\kappa_2 +(1-\alpha)b \big) \sigma_t^2,
\ealigns
which in combination with \eqref{eq:hilferuf001} yields
\baligns
\IE_{t-1} \big[ G_t  ~|~ G_t > \var{t-1}(G_t) \big]
& = a_0 \sum_{k=0}^{T-t-1} \big(  a_1\ov\kappa_2 + b \big)^k + 
\big( a_1\ov\kappa_2+ b \big)^{T-t}\sigma_t^2
= \cP^T_{t-1}\big( a_1\ov\kappa_2 + b \big).
\ealigns
This finally amounts to 
\baligns
\tildeavar{t-1}(L_T^2) &= \frac{1}{1-\alpha} \int_\alpha^1 F_Z^{-1}(u) \udu ~  \IE_{t-1} \big[ G_t  ~|~ G_t > \var{t-1}(G_t) \big]\\
&= \frac{1}{1-\alpha}\int_\alpha^1 F_Z^{-1}(u) \udu ~  \cP^T_{t-1}\big( a_1 \ov\kappa_2 + b \big),
\ealigns
which proves the assertion.
\halmos

\medskip

\noindent
{\em Proof of Proposition~\ref{prop:avar_jensen}  }\quad
A careful proof tracking reveals its similarity to the proof of Theorem~\ref{prop:consistent_avar}.
For simplicity we set $\ov\kappa= \frac{1}{1-\alpha} \, \int_{\alpha}^1 F_{Z}^{-1}(y) \ud y$ and
$\ov\kappa_2 = \frac{1}{1-\alpha} \, \int_{\alpha}^1 F_{Z^2}^{-1}(y) \ud y$.

At $t=T-1$ we have $\avarup{T-1}(L_T) = \ov\kappa \, \sigma_{T}$ which coincides with $\avar{T-1}(L_T)$. 
Since by Definition and \eqref{1} and \eqref{2},
\[
\tildeavar{T-2}(L_T) = \avar{T-2} \big( \tildeavar{T-1}(L_T) \big) = \ov\kappa\, \avar{T-2} \big( \sigma_{T} \big),
\]
we obtain
\baligns
\avar{T-2} \big( \sigma_{T} \big) &= \IE_{T-2} \Big[ \sigma_T ~|~ \sigma_T > \var{T-2}(\sigma_T) \Big]\\
&=\IE_{T-2}\Big[ \sigma_T ~|~ Z^2_{T-1} > F_{Z^2}^{-1}(\alpha)\Big].
\ealigns
by Lemma~\ref{lemma:var_lemma}.
An application of Jensen's inequality yields
\baligns
\IE_{T-2} \Big[ \sigma_T ~|~ Z^2_{T-1} > F_{Z^2}^{-1}(\alpha)\Big] & \leq \Big( \IE_{T-2} \Big[ \sigma_T^2 ~|~ Z^2_{T-1} > F_{Z^2}^{-1}(\alpha)\Big]  \Big)^{1/2}. 
\ealigns
We obtain further
\baligns
\IE_{T-2} \Big[ \sigma_T^2 ~|~ Z^2_{T-1} > F_{Z^2}^{-1}(\alpha)\Big] 
&= \IE_{T-2} \Big[ a_0 +  \sigma_{T-1}^2 (a_1 Z_{T-1}^2 + b )  ~|~ Z^2_{T-1} > F_{Z^2}^{-1}(\alpha)\Big]\\
&\quad= \frac{1}{1-\alpha}\, \int_{F_{Z^2}^{-1}(\alpha)}^\infty \Big( a_0 + \sigma_{T-1}^2\big(a_1  y + b \big) \Big) \ud F_{Z^2}(y)\\
&\quad =  a_0 + \sigma_{T-1}^2 \Big( b +a_1 \frac{1}{1-\alpha}  \int_{F_{Z^2}^{-1}(\alpha)}^\infty y \ud F_{Z^2}(y)\Big)\\
&\quad=  a_0 + \sigma_{T-1}^2 \Big( b+ a_1 \ov\kappa_2\Big),
\ealigns
which amounts to
\baligns
\tildeavar{T-2}(L_T) &\leq \ov\kappa \, \sqrt{ a_0 + \sigma_{T-1}^2 \Big(  a_1\ov\kappa_2 + b \Big) }
= \ov\kappa\, \sqrt{ \cP^T_{T-2}\Big( a_1\ov\kappa + b \Big)  }
= \avarup{T-2}(L_T).
\ealigns
This proves for $t=T-2$ that $\avarup{T-2}(L_T)$ is an upper bound for $\tildeavar{T-2}(L_T)$. 

\medskip

Now assume that $\avarup{s}(L_T) \geq \tildeavar{s}(L_T)$ holds true for $s=T-1,\ldots,t+1$. 
We show next that also 
\[
\avarup{t}(L_T) \geq \tildeavar{t}(L_T).
\]
To this end notice that 
\balign\label{eq:hilfe01}
\tildeavar{t}(L_T) = \avar{t}\big( \tildeavar{t+1}(L_T)  \big) \leq  \avar{t}\big( \avarup{t+1}(L_T)  \big).
\ealign
Moreover, we have by the induction assumption
\baligns
&\avar{t}\big( \avarup{t+1}(L_T)  \big) = \avar{t} \Big(\ov\kappa_2 \,  \sqrt{ \cP^T_{t+1}\Big( \frac{a_1}{1-\alpha} \int_\alpha^1 F_{Z^2}^{-1}(u) \udu + b \Big) } \Big)\\
&\quad = \ov\kappa_2 \,  \IE_t \Big[ \sqrt{ \cP^T_{t+1}\big(  a_1 \ov\kappa_2 + b \big) }~\Big|~ \sqrt{ \cP^T_{t+1}\big( a_1\ov\kappa_2 + b \big) }>\var{t}\Big( \sqrt{ \cP^T_{t+1}\big(  a_1 \ov\kappa_2 + b \big) } \Big) \Big].
\ealigns
By a similar calculation as in the proof of Theorem~\ref{prop:consistent_avar} and Lemma~\ref{lemma:var_lemma}, we can see that the above expression simplifies to
\baligns
&\avar{t}\big( \avarup{t+1}(L_T)  \big) = \ov\kappa_2 \,  \IE_t \Big[ \sqrt{ \cP^T_{t+1}\big(  a_1 \ov\kappa_2 + b \big) } ~\Big|~Z^2_{t+1} > F_{Z^2}^{-1}(\alpha) \Big]\\
&\quad\leq \frac{1}{1-\alpha}\int_\alpha^1 F_Z^{-1}(y) \ud y \,  \Big( \IE_t \Big[  \cP^T_{t+1}\big( a_1 \ov\kappa_2 + b \big) ~\Big|~Z^2_{t+1} > F_{Z^2}^{-1}(\alpha) \Big] \Big)^{1/2},
\ealigns
where the last line follows from Jensen's inequality. Note that
\baligns
\lefteqn{\IE_t \Big[  \cP^T_{t+1}\big( a_1 \ov\kappa_2+ b \big) ~\Big|~Z^2_{t+1} > F_{Z^2}^{-1}(\alpha) \Big]}\\
& = \frac{1}{1-\alpha} \, \int_{F_{Z^2}^{-1}(\alpha)}^\infty \Big( a_0 \sum_{k=0}^{T-t-3} \big( a_1 \ov\kappa_2+ b \big)^k
 + \Big( a_0 + \sigma_{t+1}^2\big( a_1 y + b \big) \Big) \big( a_1 \ov\kappa_2 + b \big)^{T-t-2} \Big) \ud F_{Z^2}(y) \\
& = a_0 \sum_{k=0}^{T-t-2} \big( a_1\ov\kappa_2 + b \big)^k
 + \frac{1}{1-\alpha} \big( a_1 \ov\kappa_2 + b \big)^{T-t-2} \sigma_{t+1}^2 \int_\alpha^1 \big( a_1 F^{-1}_{Z^2}(y) + b \big) \ud y\\
& = a_0 \sum_{k=0}^{T-t-2} \big( a_1 \ov\kappa_2+ b \big)^k + \sigma_{t+1}^2 \big( a_1 \ov\kappa_2+ b \big)^{T-t-1} \\
&= \cP^T_t\big(  a_1 + b \big), 
\ealigns
which implies 
\[
\avar{t}\big( \avarup{t+1}(L_T)  \big) \leq \frac1{1-\alpha} \int_\alpha^1 F_Z^{-1}(y) \ud y\, \sqrt{\cP^T_t\big(  a_1 \ov\kappa_2 + b \big)} = \avarup{t}(L_T).
\]
Finally it follows from \eqref{eq:hilfe01} that $\tildeavar{t}(L_T) \leq \avarup{t}(L_T)$.
\halmos

\bigskip 

\noindent
{\em Proof of Proposition~\ref{prop:avar_lower} }\quad
At $t=T-1$, $\avarlow{T-1}(L_T)$ coincides with $\tildeavar{T-1}(L_T)$. At $t=T-2$ we obtain
\[
\tildeavar{T-2}(L_T) = \avar{T-2} \big( \tildeavar{T-1}(L_T) \big) = \frac{1}{1-\alpha}\int_\alpha^1 F_Z^{-1}(y) \ud y \, \avar{T-2} \big( \sigma_{T} \big).
\] 
For simplicity we write $\ov\kappa = \frac{1}{1-\alpha}\int_\alpha^1 F_Z^{-1}(y) \ud y$.
By continuity of the distribution function inherited from $Z$,
\baligns
\avar{T-2} \big( \sigma_{T} \big) &= \IE_{T-2}\Big[ \sigma_T ~|~ \sigma_T > \var{T-2}(\sigma_T)\Big],
\ealigns
which by Lemma~\ref{lemma:var_lemma} rewrites as
\baligns
\avar{T-2} \big( \sigma_{T} \big) &= \IE_{T-2}\Big[ \sigma_T ~|~ Z^2_{T-1} > F_{Z^2}^{-1}(\alpha)\Big].
\ealigns
Using the $\cF_{T-2}$-measurability of $\sigma_{T-1}$ and $a_0>0$, we calculate further,
\balign\label{b2}
\IE_{T-2} \Big[ \sigma_T ~|~ Z_{T-1} > F_Z^{-1}(\alpha)\Big] 
&=  \IE_{T-2} \Big[ \sqrt{ a_0 + \sigma_{T-1}^2\big( a_1 Z_{T-1}^2  + b \big) } ~|~ Z^2_{T-1} > F_{Z^2}^{-1}(\alpha)\Big]\nonumber\\
&\geq \sigma_{T-1}\, \IE_{T-2} \Big[ \sqrt{ a_1 Z_{T-1}^2  + b } ~|~ Z^2_{T-1} > F_{Z^2}^{-1}(\alpha)\Big]\nonumber\\
&= \sigma_{T-1}  \,\frac{1}{1-\alpha}\int_{F_{Z^2}^{-1}(\alpha)}^\infty \sqrt{a_1 y + b} \, \ud F_{Z^2}(y).
\ealign
Hence, it follows that
\baligns
\tildeavar{T-2}(L_T) &\geq \ov\kappa \, 
\frac{1}{1-\alpha}\int_{\alpha}^1 \sqrt{a_1 F_{Z^2}^{-1}(y) + b} \, \ud y \, \sigma_{T-1} 
= \avarlow{T-2}(L_T).
\ealigns
This proves for $t=T-2$ that $\avarlow{T-2}(L_T)$ as in \eqref{eq:avar_lower} is a lower bound for 
$\tildeavar{T-2}(L_T)$.

Now assume that $\avarlow{s}(L_T) \leq \tildeavar{s}(L_T)$ holds true for $s=T-1,\ldots,t+1$. 
We show next that also
\[
\avarlow{t}(L_T) \leq \tildeavar{t}(L_T).
\]
To this end notice that 
\balign\label{eq:hilfe02}
\tildeavar{t}(L_T) = \avar{t}\big( \tildeavar{t+1}(L_T)  \big) \geq  \avar{t}\big( \avarlow{t+1}(L_T)  \big).
\ealign
Moreover, we have by the induction assumption
\baligns
&\avar{t}\big( \avarlow{t+1}(L_T)  \big) 
= \avar{t} \Big(\ov\kappa \,  \sigma_{T-1}\frac1{1-\al)} \int_\al^1 \sqrt{a_1F_{Z^2}^{-1}(y)+b} \ud y\Big)\\ 
&\quad = \ov\kappa \, \Big( \frac{1}{1-\alpha}  \int_{F_{Z^2}^{-1}(\alpha)}^\infty \sqrt{a_1 y + b} \, \ud F_{Z^2}(y)  \Big)^{T-t-2} \IE_t \Big[ \sigma_{t+2} ~|~ \sigma_{t+2} > \var{t}\big( \sigma_{t+2} \big) \Big]\\
&\quad = \ov\kappa \, \Big( \frac{1}{1-\alpha}  \int_{F_{Z^2}^{-1}(\alpha)}^\infty \sqrt{a_1 y + b} \, \ud F_{Z^2}(y) \Big)^{T-t-2} \IE_t \Big[ \sigma_{t+2} ~|~ Z^2_{t+1} > F_{Z^2}^{-1}(\alpha) \big)\Big],
\ealigns
where the last equality follows from Lemma~\ref{lemma:var_lemma}. 
Then by the same calculation which lead to \eqref{b2},
\baligns
\IE_t \Big[ \sigma_{t+2} ~|~ Z^2_{t+1} > F_{Z^2}^{-1}(\alpha) \big)\Big] 
&= \sigma_{t+1}  \,\int_{\alpha}^1 \sqrt{a_1 F^{-1}_{Z^2}(y) + b} \, \ud y,
\ealigns
this yields together with \eqref{eq:hilfe02},
\[
\tildeavar{t}(L_T) \geq \ov\kappa \, \left( \frac{ 1}{1-\alpha}\int_{\alpha}^1 \sqrt{a_1 F^{-1}_{Z^2}(y) + b} \, \ud y \right)^{T-t-1} \sigma_{t+1} = \avarlow{t}(L_T).
\]
\halmos

\subsubsection*{Acknowledgements}
We are grateful to one of the Reviewers and Marcin Pitera, who pointed out some errors in a previous version of this paper.

\bibliography{references}
\bibliographystyle{abbrvnat}                        
\end{document}